\documentclass{IEEEtran}
\pdfoutput=1
\usepackage{cite}
\usepackage{color}
\usepackage{graphicx}
\usepackage{epstopdf}
\usepackage[cmex10]{amsmath}
\usepackage{amssymb}
\usepackage{algorithm,algpseudocode}

\usepackage{amsmath,amssymb}
\usepackage{lipsum}
\usepackage{comment}
\usepackage{array}
\input{mysymbol.sty}
\usepackage{needspace}

\usepackage{amsthm,bm}
\usepackage{bbm}
\usepackage{url}
\usepackage{enumerate}

\usepackage{subfigure}
\usepackage{tikz}
\usetikzlibrary{arrows,arrows.meta}
\usepackage{pgfplots}
\usepgfplotslibrary{groupplots}
\usepgfplotslibrary{fillbetween}
\newtheorem{theorem}{Theorem}
\newtheorem{definition}{Definition}
\newtheorem{lemma}{Lemma}

\theoremstyle{definition}

\newtheorem{remark}{Remark}

\DeclareMathOperator{\subjectto}{s.t.}
\def\nnil{\nil}
\newcounter{prob}

\author{Santiago Paternain$^\dagger$, Miguel Calvo-Fullana$^\star$, Luiz F. O. Chamon$^\S$ and Alejandro Ribeiro$^\ddagger$
\thanks{Work supported by ARL DCIST CRA W911NF-17-2-0181 and the Intel Science and Technology Center for Wireless Autonomous Systems. $^\dagger$Electrical Computer and Systems Engineering, Renssealaer Polytechnic Institute. Email: paters@rpi.edu. $^\star$Department of Aeronautics and Astronautics, Massachusetts Institute of Technology. Email: cfullana@mit.edu. $^\S$ University of California, Berkeley. Email: lfochamon@berkeley.edu.
$^\ddagger$ Department of Electrical and Systems Engineering, University of Pennsylvania. Email: aribeiro@seas.upenn.edu.
}}

\renewcommand{\comment}[1]{}
\newcolumntype{S}{>{\centering\arraybackslash} m{.10\linewidth} }
\newcolumntype{T}{>{\centering\arraybackslash} m{.30\linewidth} }
\title{Safe Policies for Reinforcement Learning via Primal-Dual Methods}
\begin{document}

\maketitle

\begin{abstract}
  In this paper, we study the design of controllers in the context of stochastic optimal control under the assumption that the model of the system is not available. This is, we aim to control a Markov Decision Process (MDP) of which we do not know the transition probabilities, but we have access to sample trajectories through experience. We define safety as the agent remaining in a desired safe set with high probability during the operation time. The drawbacks of this formulation are twofold. The problem is non-convex and computing the gradients of the constraints with respect to the policies is prohibitive. Hence, we propose an ergodic relaxation of the constraints with the following advantages. (i) The safety guarantees are maintained in the case of episodic tasks and they hold until a given time horizon for continuing tasks. (ii) The constrained optimization problem despite its non-convexity has arbitrarily small duality gap if the parametrization of the controller is rich enough. (iii) The gradients of the Lagrangian associated to the safe learning problem can be computed using standard Reinforcement Learning (RL) results and stochastic approximation tools. Leveraging these advantages, we exploit primal-dual algorithms to find policies that are safe and optimal. We test the proposed approach in a navigation task in a continuous domain. The numerical results show that our algorithm is capable of dynamically adapting the policy to the environment and the required safety levels.
\end{abstract}

%!TEX root = root.tex

%The goal of the agent is thus to find a control law that maximizes the value function of the MPD, i.e., the expectation of the cumulative rewards \cite{sutton1998reinforcement}.

\section{Introduction}\label{sec_intro}

Safety is a fundamental feature in the control design for physical systems. Indeed, controllers for power systems are designed to avoid voltage instabilities~\cite{van2007voltage} which result in unsafe operation. Likewise, in the case of robotics, collision avoidance~\cite{khatib1985real} %both for single \cite{khatib1985real} and multi-agent systems \cite{mastellone2008formation}
 needs to be guaranteed for their correct operation and to guarantee the integrity of people and property that may surround the robots.

\subsection{Safe reinforcement learning}\label{sec_intro_a}

In this work we consider the Reinforcement Learning (RL) framework understood as the problem of controlling an unknown Markov Decision Process~(MDP)~\cite{howard1964dynamic}.  When the transitions probabilities of the MDP are available, optimal control laws---or \emph{policies}---can be obtained for these processes using dynamic programming~\cite{bertsekas1996neuro}. In contrast, when the underlying MDP is unknown or the system is too complex, the policy needs to be learned from samples of the system. Typically, this is done by assigning an instantaneous \emph{reward} to the system actions that describes the task to be learned. These rewards can be aggregated over a trajectory to determine cumulative rewards. Since the instantaneous rewards depend on the state and on the actions selected, cumulative rewards are a measure of the quality of the decision-making policy of the agent. The objective of the agent is therefore to find a policy that maximizes the expectation of the cumulative rewards, which is known as the \emph{value function}~(or the \emph{Q-function}) of the MDP~\cite{sutton2018reinforcement}. Solutions to these problems can be roughly divided among those that learn the Q-function~\cite{watkins1992q}, and those that attempt to run gradient ascent in the space of policies~\cite{williams1992simple,sutton2000policy}.

%Solutions to these problems can be roughly divided among those that learn the Q-function to then choose for any given state the action that maximizes the function \cite{watkins1992q}, and those that attempt to directly learn the optimal policy \cite{williams1992simple,sutton2000policy}. A drawback of Q-learning \cite{watkins1992q} for continuous state-action spaces\textemdash and any other algorithm that learns Q-functions for that matter\textemdash is that even after training, the selection of each action amounts to solving a nonconvex optimization problem \blue{(unless there is special structure that can be exploited in the problem, e.g. a linear system with quadratic costs).}
%This motivates the development of algorithms that attempt to learn the optimal policy directly by performing (stochastic) gradient ascent on the value function with respect to a policy variable \cite{williams1992simple, sutton2000policy}.

A notable drawback of these methods is that they are not always suitable for dangerous or risky tasks~\cite{heger1994consideration, mihatsch2002risk, geibel2005risk}. %Indeed, many applications require robust control strategies which also take into account, for instance, the variance of the accumulated reward to avoid situations in which its value on a specific realization of the process is considerably worse than its mean. %Consider the case of a self-driving car deployed in an urban environment. To reach a destination as fast as possible, the optimal policy may be such that it makes risky maneuvers, such as driving close to other cars or crossing pedestrians. Due to the random components in the vehicle actions and the behavior of other cars and pedestrians, collision avoidance, in general, cannot be guaranteed with probability one.
To promote safe behaviors, it is common to include explicit chance constraints~\cite{geibel2006reinforcement, kadota2006discounted, delage2010percentile, di2012policy, hutter2002self, chow2017risk}. These constraints are probabilistic in nature, in the sense that they mandate certain requirements to hold with some given minimum probability. These requirements can involve, for instance, lower bounds on the value function or additional value functions~\cite{kadota2006discounted, geibel2006reinforcement, delage2010percentile}, or arbitrary functions of the state-action space~\cite{chow2017risk, achiam2017constrained}. %In the case of finite state-action spaces, these constrained learning problems can be written as linear programs on the occupancy measure~\cite{hutter2002self, geibel2006reinforcement} and in some cases the policy can be recovered. Other approaches leverage approximate trust region methods~\cite{achiam2017constrained}, or the application of primal-dual algorithms~\cite{chow2017risk}.
  {The main limitation of these approaches is that they are not guaranteed to be safe in control theoretic terms. For instance, they cannot guarantee positive invariance under closed loop which is a common property in controls~\cite{wieland2007constructive,koditschek1990robot,mayne2000constrained}. }
%\cite{di2012policy} Mean variance trade-off

In this work, we formulate safety constraints by imposing a desired lower bound on the probability of the trajectory to be contained in a safe set, akin to a probabilistic version of positive invariance (see e.g., \cite{khalil2002nonlinear}). This notion of safety is stronger than that in \cite{paternain2019learning} where the probability of remaining in the safe set is for each state of the trajectory. The challenge in RL is that these probabilities cannot be computed since the transitions are assumed unknown. Thus, they can only be estimated through experiments. %Yet, there is no explicit expression that relates these probabilities with the policy that can in practice be used to find a safe policy.
Yet, there is no expression that allows for an update on the policy that takes these probabilities into account. Following the ideas in \cite{paternain2019learning} we propose a relaxation of the definition of safety~(Section~\ref{sec_safe_learning}) and provide guarantees on the ergodic safety of policies learned using the relaxed formulations~(Section~\ref{sec_safety_guarantees}). Namely, we show that these relaxations preserve the safety guarantees.

The main advantage of these relaxations is that they yield an expression that allows the computation of stochastic gradients of the constraints with respect to the policy. Thus, allowing the computation of risk aware (regularized) solutions of the constrained problem. This is a common approach to constrained optimization that has been applied in RL~\cite{howard1972risk, geibel2005risk}. However, the relationship between  regularization coefficient and constraint satisfaction is not straightforward. Indeed, the selection of these hyperparameters, in general, results in a time-consuming process that is domain dependent and often requires domain knowledge~\cite{tessler2018reward,leike2017ai,mania2018simple}. To overcome this limitation, we propose to find the set of parameters in a principled manner by solving the dual problem associated to the constrained problem. Yet, the lack of convexity of the problem prevents us, in principle, to claim any optimality in the selection of such coefficients. In general, the dual optimum is only an upper bound on the primal problem (weak duality). It is well known that if the primal problem is convex and Slater constraint qualifications hold, then the solution of the dual problem is also a lower bound on the primal problem (see e.g.~\cite[Chapter 5]{boyd2004convex}). This is known as strong duality. In this work, we show that despite the non-convexity of the primal problem, strong duality holds (Section~\ref{sec_zdg}). {Strong duality of a class of constrained reinforcement learning problems has been established in \cite{cPaternainEtal2019d}. In this work, we provide a new proof based on geometric arguments instead of the analysis of perturbations of mathematical programs.}

%this approach does not guarantee that the solutions satisfy the constraints. To that end, we establish that the safe-learning problem, despite its non convexity has zero duality gap (Section~\ref{sec_zdg}). The latter implies that the trade-offs expressed by different weights in the risk aware rewards are the same as those expressed by the probability specifications, in the sense that they trace the same Pareto front. Nevertheless, the relationship between weights and specifications is not trivial and specifying the constrained problem is often considerably simpler. \blue{The zero duality gap of a class of constrained reinforcement learning problems has been established in \red{[REF]}. In this work, we provide a new proof based on a geometrical analysis instead of using analyzing the perturbation function as in \red{[REF]}. }

 Other than concluding remarks, the paper finishes with numerical experiments in which we show that primal-dual methods (described in Section \ref{sec_pd}) can automatically adjust the trade-off between goal and safety~(Section~\ref{num_results}). Before formalizing the positively invariant reinforcement learning problem, in Section \ref{sec_safe_control} we relate our work to approaches to safe control in the literature.

\subsection{Alternative approaches to safety in control }\label{sec_safe_control}
Several control techniques exist to guarantee safety in terms of positive invariant sets. Some include control barrier functions (CBFs) \cite{wieland2007constructive}, artificial potentials \cite{koditschek1990robot}, or explicit state and input constraints for instance in the context of Model Predictive Control (MPC) \cite{mayne2000constrained}. These in general rely on the knowledge of the dynamical system under consideration. Indeed, CBFs define an admissible control space by ensuring that the Lie derivative of the CBFs remains positive in the safe set. In the case of artificial potentials, in general kinematic systems are considered. This is equivalent to having available a low-level controller that takes the dynamics of the system into account. When systems with dynamics are considered, the input of the system needs to include a dissipative term \cite{koditschek1991control} that depends on the dynamics.  In the case of MPC, the inputs are designed by solving an optimization problem that depends explicitly on the model of the system. If the dynamics are not known, one can attempt to use system identification~ \cite{deistler1995consistency,tsiamis2019finite,dean2019sample}, however most of the guarantees in terms of consistency are only provided in the case of linear systems.

%Designing controllers for safe operation in the absence of a model is a more challenging problem, especially in the case of stochastic systems. A possible avenue to pursue is to perform stochastic system identification \cite{van2012subspace} and then resort to robust versions of the techniques described before to take into account the errors in the identification process. In the case of linear systems classical results study the asymptotic consistency of stochastic subspace system identification \cite{deistler1995consistency} and contemporary results focus on the finite sample bounds \cite{tsiamis2019finite,dean2019sample}.

%In the context of nonlinear dynamical systems one can do system identification using Volterra series methods, however in general they require excessively large amounts of data \cite{billings1980identification}. To overcome this issue one can consider block-structured systems \cite{giri2010block,schoukens2017identification}. These methods, however can be applied to a very special form of model and usually its form has to be known prior to identification. Neural Networks are also a popular approach for system identification but in general they lack theoretical guarantees and interpretability (see e.g., \cite{andersson2019deep}).

 The shortcomings of these methods in the case of unknown stochastic non-linear systems makes the study of RL solutions relevant. We discuss the main strategies for safe RL other than the risk-aware rewards already mentioned in Section \ref{sec_intro_a}. A comprehensive review of this topic can be found in~\cite{garcia2015comprehensive}. One approach consists in formulating robust problems in which the policy is optimized over its worst-case return~\cite{heger1994consideration, coraluppi1999risk}. However, these techniques generally yield policies too conservative for the average scenario. % and make it hard to control the trade-off between safety and performance.
%
%Although this approach makes the risk-performance trade-off more transparent, it requires this balance to be hand-tuned, an often time consuming and challenging task that requires application\textemdash and domain-specific expert knowledge, as showed in \cite{tessler2018reward,leike2017ai,mania2018simple}. Moreover, implicit interference between the goals may lead to training plateaus as they compete for resources in the policy~\cite{schaul2019ray}. What is more, the function of the reward is to inform the goal of the agent, not prior knowledge on how to complete it. Indeed, ``the reward signal is your way of communicating to the robot \emph{what} you want it to achieve, not \emph{how} you want it achieved''~\cite[Section~3.2]{sutton2018reinforcement}. This formulation is closely related to the method proposed here.
%
The second approach addresses this issue by modifying the learning procedure instead of the reward. By performing safe exploration~\cite{moldovan2012safe,turchetta2016safe,koller2018learning}, the agent learns from safe trajectories and is therefore biased to learn safe policies. These methods require, not only the existence of a safe policy, but it also needs to be available to the agent.

%\input{introduction}
%!TEX root = root.tex

\section{Problem Formulation}\label{sec_problem_formulation}
In this work we consider a stochastic optimal control problem under safety constraints and we assume that the model is unknown. Formally, let~$\ccalS$ and~$\ccalA$ be compact sets describing the states and actions of the agent, respectively. Denote by $t\in\mathbb{N}$ the time index. We consider a random controller\textemdash or \emph{policy}\textemdash defined as a conditional distribution~$\pi_\theta(a|s)$ from which the agent draws actions~$a \in \ccalA$ when in state~$s \in \ccalS$. We restrict our attention to controllers parametrized by a vector~$\theta \in \mathbb{R}^d$. The action selected by the agent drives it to another state according to the transition dynamics of the system defined by the conditional probability~$\mathbb{P}_{s_t\to s_{t+1}}^{a_t}(\tilde{\ccalS}) := \mathbb{P}(s_{t+1}\in \tilde{\ccalS} \mid s_t,a_t)$, where $\tilde{\ccalS}\subset \ccalS$,~$t \in \mathbb{N}$, $s_t \in \ccalS$, and~$a_t \in \ccalA$. This process is assumed to satisfy the Markov property~$\mathbb{P}(s_{t+1} \in \tilde{\ccalS} \mid (s_u,a_u),\, \forall u \leq t) = \mathbb{P}(s_{t+1} \in \tilde{\ccalS} \mid s_t,a_t)$, hence it is termed an MDP.  {The initial state is randomized according to a probability distribution $\mathbb{P}(s_0\in\tilde{\ccalS})$}. In addition, the action selected provides a reward to the agent from the function~$r:\ccalS \times \ccalA \to \mathbb{R}$ that informs the agent of the quality of the decision.

 The goal of the agent is to find a parametrization~$\theta$ of the policy that maximizes the value function of the MDP. This is defined as the expected value of the discounted cumulative rewards obtained along a trajectory
\begin{equation}\label{eqn_value_function_discounted}
	V(\theta) = \mathbb{E}_{\bba\sim\pi_{\theta}(\bba|\bbs)}
		\left[\sum_{t=0}^\infty \gamma^t r(s_t,a_t)\right],
\end{equation}
where~$\bbs$ and $\bba$ represent the trajectory of the system and~$\gamma \in (0,1]$ is the discount factor. The parameter $\gamma$ defines how myopic the agent is. Indeed, the smaller~$\gamma$ is, the faster the geometric series vanishes which implies that less importance is placed in future rewards relative to an agent that selects $\gamma$ closer to $1$.  In the limit, where $\gamma=1$ all rewards have the same importance. However, in this case we assume that all rewards after a finite horizon time $T$ are zero. {Although the finite horizon and the infinite horizon capture different operation principles it is possible to show that the latter reduces to the former when the horizon is selected randomly~(see Remark~\ref{remark_equivalence}).}

As we argued in Section~\ref{sec_intro}, maximizing $V(\theta)$ in~\eqref{eqn_value_function_discounted} may lead to unsafe or risky policies. To formalize the notion of safety, let~$\ccalS_{0} \subset \ccalS$ denote a set of states in which the agent is required to remain. Then, we define safety as a probabilistic extension of control invariant sets \cite{khalil2002nonlinear}.
\begin{definition}\label{def_safe_policy}
	We say a policy~$\pi_\theta$ is $(1-\delta)$-safe for the set~$\ccalS_0 \subset \ccalS$ if~$\mathbb{P}\left( \bigcap_{t \geq 0} \{s_{t} \in \ccalS_0\} \mid \pi_\theta\right) \geq 1-\delta$.
\end{definition}
In other words, a policy is safe if the trajectories it generates remain within the safe set~$\ccalS_0$ with high probability. Note that this is a stricter version of safety than the one used in~\cite{paternain2019learning} where each state of the trajectory was considered separately, i.e., where a policy was considered safe if~$ \mathbb{P}( s_{t} \in \ccalS_0 \mid \pi_\theta) \geq 1-\delta $ for all $t\geq  0$. %\blue{The latter is equivalent to the notion of Probabilistic Invariant Sets \blue{[REF]} which is a milder requirement \red{as we show in Proposition} }

In general, different requirements impose different safety guarantees, we consider a collection of $m$ different subsets~$\ccalS_1, \ldots, \ccalS_m \subset \ccalS$ of the state-space and we can summarize the problem of interest as that of solving the following constrained optimal control problem.
\begin{equation}\label{eqn_optimization_problem}
\begin{aligned}
	\underset{\theta \in \mathbb{R}^d}{\max} & &&V(\theta)
	%\text{maximize}_{\theta \in \ccalH}& &&V_{T/\infty}(\theta)
	\\
	\subjectto &
		&&\mathbb{P}\left( \bigcap_{t = 0}^{\infty} \{ s_{t} \in \ccalS_i\} \,\bigg\vert\, \pi_\theta \right) \geq 1-\delta_i
		\text{, } i = 1, \ldots, m
		\text{.}
\end{aligned}
\end{equation}
Since we assume that the transition probabilities of the MDP are not available to the agent, the probability~$\mathbb{P}(s_{t} \in \ccalS_i)$ can only be evaluated by experience which prevents us from establishing a relation between~$\theta$~(i.e., the policy) and the constraint in~\eqref{eqn_optimization_problem}. %\red{ This in turn results in the impossibility of modifying the policy so to satisfy the constraints at all times.  A possibility to overcome this limitation is to integrate prior knowledge about the system into the decision making process by projecting the action selected into a set that ensures the satisfaction of the constraints~\cite{dalal2018safe}. Such set is constructed based on previous transitions that have been observed and as such it has the disadvantage that safety is not guaranteed unless the agent operates in a state in the neighborhood of previously observed ones. In addition, such projection might result in operation that is not optimal.}

A common alternative to overcome this issue is to modify the reward function in~\eqref{eqn_value_function_discounted} to make it risk-aware~\cite{howard1972risk, geibel2005risk}, i.e., to use a reward of the form
\begin{equation}\label{eqn_risk_aware_reward}
	r_{w}(s_t,a_t) = {r}(s_t,a_t) + \sum_{i=1}^m w_i \mathbbm{1}(s_t \in \ccalS_i)
		\text{,}
\end{equation}
where~${r}$ is the original reward function describing the agent task, $w_i > 0$ are safety-related rewards, and the indicator function is such that~$\mathbbm{1}(s_t \in \ccalS_i) = 1$ if~$s_t \in \ccalS_i$ and zero otherwise. In other words, the agent receives an extra reward of~$w_i$ for respecting the $i$-th safety specifications. Since this approach amounts simply to modifying the reward function, common RL techniques used to maximize~$V(\theta)$ apply to this setting~\cite{sutton2018reinforcement}. Nevertheless, selecting parameters~$w_i$ that lead to $(1-\delta)$-safe policies is challenging given that there is no straightforward relation between the~$w_i$ and the probabilities~$\mathbb{P}(s_t \in \ccalS_i)$. Moreover, not only do their values depend on~${r}$, but they must also strike a balance between safety and task completion: large values of~$w_i$ can lead to policies that are safe \emph{because} they do not achieve the goal \cite{paternain2019learning}.

In the sequel, we leverage duality and probabilistic inequalities to put forward a relaxation of~\eqref{eqn_optimization_problem} that leads to guaranteed~$(1-\delta)$-safe policies. The advantages of the proposed method are twofold. The relaxation is such that it results in an expression similar to~\eqref{eqn_risk_aware_reward}, allowing commonly used RL algorithms to be applied  directly. Moreover, it provides a systematic way of adapting the coefficients~$w_i$ to obtain safe policies. It is worth pointing out that the relaxation and the aforementioned advantages hold for any stochastic control problem of the form \eqref{eqn_optimization_problem} regardless of whether the model is available or not. 

Before proceeding, we present pertinent remarks regarding the process of designing the controller and the relationship between episodic and discounted continuing tasks.
{
\begin{remark}
The problem of interest is that of designing a controller that is safe and we do not guarantee that during the training process the constraints in~\eqref{eqn_optimization_problem} are satisfied. This is a problem of learning \emph{safety}, as compared to learning \emph{safely} (safe exploration), which we do not consider.
\end{remark}
}
\begin{remark}\label{remark_equivalence}
  In this remark we discuss the relationship between the undiscounted ($\gamma=1)$ finite horizon problem and the discounted $\gamma<1$ with infinite horizon. {In particular, we discuss that the infinite
  	horizon problem can be reduced to a finite horizon problem with a
  	random stopping time.} This discussion is inspired in \cite[Section 2.3]{bertsekas1996neuro} and in the proofs of \cite[Proposition 2 and 3]{paternain2018stochastic}. Let us start by considering the finite horizon value function with a horizon chosen from a geometric distribution with parameter $\gamma\in(0,1)$. Then, one has that 
  \begin{equation}
    \E{\sum_{t=0}^Tr(s_t,a_t)} = \E{\sum_{t=0}^\infty \mathbbm{1}(t\leq T)r(s_t,a_t)}.
    \end{equation}
  Under mild assumptions on the reward function it is possible to exchange the sum and the expectation (see e.g., \cite[Proposition 2]{paternain2018stochastic} ). Also, assuming that the horizon is drawn independently from the trajectory, we can write $\E{\mathbbm{1}(t\leq T)r(s_t,a_t)} = \E{\mathbbm{1}(t\leq T)}\E{r(s_t,a_t)}$. This yields
    \begin{equation}\label{eqn_remark}
    \E{\sum_{t=0}^Tr(s_t,a_t)} = \sum_{t=0}^\infty \E{ \mathbbm{1}(t\leq T)}\E{r(s_t,a_t)}.
    \end{equation}
    Further notice that the expectation of the indicator function is the probability of $t$ being less than $T$. Since $T$ is drawn from a geometric distribution it follows that $\E{ \mathbbm{1}(t\leq T)} = \gamma^t(1-\gamma)$. Thus, \eqref{eqn_remark} reduces to
    \begin{equation}
    \E{\sum_{t=0}^Tr(s_t,a_t)} = (1-\gamma)\sum_{t=0}^\infty \gamma^t\E{r(s_t,a_t)}.
        \end{equation}
{Exchanging back the expectation and the sum establishes the that objective function of the infinite horizon problem is the same as that of a finite horizon problem with a random stopping time.}
        %
%    \begin{equation}
%    \E{\sum_{t=0}^Tr(s_t,a_t)} = (1-\gamma)\E{\sum_{t=0}^\infty \gamma^tr(s_t,a_t)}.
%        \end{equation}
    %
%    Which establishes the equivalence between the two formulations. 
  \end{remark}

%!TEX root = root.tex

\section{Learning Safe Policies}\label{sec_safe_learning}

If the transition probabilities of the system were known, \eqref{eqn_optimization_problem} could be solved by directly imposing constraints on the probabilities. For instance, using MPC~\cite{mayne2000constrained}. However, this is not the scenario in RL problems, where the transition probabilities can only be evaluated through experience. Hence, although the safety probabilities in~\eqref{eqn_optimization_problem} can be estimated, there is no straightforward way to optimize~$\pi_\theta$ with respect to them. To overcome this difficulty, we relax the chance constraints in~\eqref{eqn_optimization_problem} in the form of the cumulative costs in~\eqref{eqn_value_function_discounted}. Explicitly, define
\begin{align}
%	U_{T,i}(\theta) &= \frac{1}{T+1}\sum_{t=0}^{T} P(s_t\in\ccalS_i|\pi_\theta)
%		\label{eqn_average_prob}
  %		\text{,} 
%%	\\
	U_{i}(\theta) &= \sum_{t=0}^\infty \gamma^t \mathbb{P}(s_t \in \ccalS_i|\pi_\theta)
		\label{eqn_discounted_prob}
		\text{.}
\end{align}

Notice that in the case of $\gamma=1$ and finite horizon the relaxation proposed is equivalent to satisfying the chance constraints in average. This relates to the idea of online learning~\cite{vapnik2000nature,paternain2016online,chen2017online}. %Moreover, in view of the equivalence between the finite and infinite horizon case discussed in Remark~\ref{remark_equivalence}, the proposed relaxations are equivalent when the horizon~$T$ is drawn randomly from a geometric distribution. 
With this definition, we can formulate the following optimal control problem 
\begin{equation}\label{P:parametric}
  \begin{aligned}
	P^\star_\theta\triangleq\max_{\theta\in\mathbb{R}^d}& &&V(\theta)\triangleq	\mathbb{E} \left[ \sum_{t = 0}^\infty \gamma^t r(s_t,a_t) \mid \pi_\theta \right]
	\\
	\subjectto& &&U_i(\theta)  \geq c_i, \, \mbox{for all}\, i=1,\ldots m .
  \end{aligned}
  \end{equation}
Although $(1-\delta_i)$-safe policies guarantee that~$U_{i}(\theta) > \alpha(1-\delta_i)$, with $\alpha>0$, these are not sufficient conditions to achieve safety as defined in Definition~\ref{def_safe_policy}. Hence, to guarantee the desired levels of safety we need to define appropriate values~$c_i$ in \eqref{P:parametric}. In Section \ref{sec_safety_guarantees} we provide theoretical guarantees on these values to guarantee $(1-\delta_i)$-safety. {These values depend exclusively on the desired level of safety and a time horizon until which safety can be guaranteed and can therefore readily be computed. Hence, they are not hyperparameters.}

Notice that in general Problem \eqref{P:parametric} is a non-convex constrained optimization problem (see e.g.~\cite{ding2020natural}).  Given the challenges that solving constrained non-convex problems entail, a common workaround to compute approximate solutions to \eqref{P:parametric} is to solve its dual relaxation. To that end define the Lagrangian associated with \eqref{P:parametric} 
\begin{equation}\label{eqn_lagrangian}
\ccalL(\theta,\lambda)\triangleq  V(\theta) +\sum_{i=1}^m\lambda_i\left(U_i(\theta)-c_i\right), 
  \end{equation}
where $\lambda\in\mathbb{R}^m_+$ are the multipliers associated to the constraints. Then, the dual function is defined as 
\begin{equation}\label{eqn_parametric_dual_function}
d_\theta(\lambda)\triangleq \max_{\theta\in\mathbb{R}^d}\ccalL(\theta,\lambda).
  \end{equation}
The dual function provides an upper bound on the problem \eqref{P:parametric} for any $\lambda$, i.e., $P_\theta^\star\leq d_\theta(\lambda)$ \cite[Chapter 5]{boyd2004convex}. Hence in general, one is interested in finding the $\lambda$ that provides the tightest of the upper bounds. This defines the dual problem 
  \begin{equation}\label{P:dual_parametric}
  \begin{aligned}
	D_\theta^\star \triangleq\min_{{\lambda} \in \mathbb{R}^m_+}  d_\theta(\lambda). 
  \end{aligned}
  \end{equation}

    Notice that the dual function is the point-wise maximum of linear functions, and it is therefore convex \cite[Section 3.2.3]{boyd2004convex} regardless of the structure of problem \eqref{P:parametric}. Hence, the dual problem can be solved efficiently using gradient-like algorithms. We will exploit this fact in Section~\ref{sec_pd}. 

%

%In addition, solving \eqref{P:dual_parametric} provides a principled way of selecting the multipliers as compared to solving the optimal control problem with the risk aware reward defined in \eqref{eqn_risk_aware_reward}, where the coefficients are hyper-parameters.

In what follows we compare the approach proposed here, i.e., solving \eqref{P:dual_parametric} with its regularized counterpart \eqref{eqn_risk_aware_reward}. To make this comparison explicit, observe that by definition of expectation we can write the probabilities $\mathbb{P}(s_t\in\ccalS_i\mid \pi_\theta) = \mathbb{E}\left[\mathbbm{1}(s_t\in \ccalS_i)\mid \pi_\theta\right]$. Hence, we write the constraints of problem \eqref{P:parametric} as 
\begin{equation}
U_{i}(\theta)  =  \sum_{t=0}^\infty \gamma^t \mathbb{P}\left(s_t\in\ccalS_i\mid \pi_\theta\right)=\mathbb{E}\left[\sum_{t=0}^\infty\gamma^t\mathbbm{1}\left(s_t\in\ccalS_i\right)\big| \pi_\theta\right],
\end{equation}
where the sum and the expectation can be exchanged using the Monotone Convergence Theorem (see e.g.,\cite{durrett2010probability}). 

With this observation the problem of computing the dual function (cf., \eqref{eqn_parametric_dual_function}), i.e., maximizing the Lagrangian for a given $\lambda$ is equivalent to solving the unconstrained optimal control problem with the following reward
\begin{align}\label{reward_constrained}
  r_\lambda(s,a) =r(s,a) + \sum_{i=1}^m\lambda_i(\mathbbm{1}(s\in\ccalS_i) -c_i).
  \end{align}
Indeed, with this definition it follows that the Lagrangian is equivalent to 
\begin{equation}\label{eqn_parametric_lagrangian}
\ccalL(\theta,\lambda) = \mathbb{E}\left[\sum_{t=0}^\infty\gamma^t r_\lambda(s,a)\mid \pi_\theta\right]. 
  \end{equation}
In the case of model free approaches Problem \eqref{eqn_parametric_lagrangian} can therefore be solved with classic RL algorithms such as policy gradient \cite{williams1992simple,sutton2000policy,paternain2018stochastic} or actor-critic methods \cite{konda2000actor}. If the transition probabilities where known the same ideas apply as well, where the maximization can be done using dynamic programming techniques \cite{bertsekas1995dynamic}. 

Regardless of the method selected to maximize the Lagrangian and whether the model is known, formulations of the form \eqref{P:dual_parametric} have an advantage over considering the risk aware formulation~\eqref{eqn_risk_aware_reward}. The computational overhead of selecting the parameter $\lambda$ is that of solving a convex optimization problem which is more efficient than exhaustive grid search or bisection for example.

Another advantage of solving \eqref{P:dual_parametric} is that for this class of problems we can establish that the duality gap $P_\theta^\star-D_\theta^\star$ can be made arbitrarily small (cf., Theorem \ref{T:parametrization}). The fact that under mild conditions the duality gap for convex problems is zero is a well-known fact (see e.g., \cite[Chapter 5]{boyd2004convex}). However, despite problem \eqref{P:parametric} not being convex we can provide similar guarantees.

With these advantages in mind, in the next section we establish the values of the slacks $c_i$ that are required to guarantee safety of the controllers designed solving \eqref{P:parametric}. As previously mentioned, the slacks depend only on the desired safety level $\delta_i$ and on the horizon of safe operation.

\section{Safety Guarantees}\label{sec_safety_guarantees}
In this section we establish the values of the parameters $c_i$ in problem~\eqref{P:parametric} that guarantee that any solution of the optimization problem is safe (cf., Definition \ref{def_safe_policy}). To do so, we rely on the following technical lemma.

\begin{lemma}\label{T:prob_inequality}
Consider a non-increasing sequence of non-negative elements~$\mu_t \geq 0$ such that $0< \sum_{t=0}^\infty \mu_t <\infty$ and events~$\ccalE_t$ for~$t \geq 0$. Then, for any~$\delta>0$,~$\mu^\prime>0$ and~$T$ such that~$\mu^\prime \leq \mu_T$ it holds that
\begin{equation}\label{eqn_prob_inequality}
  \sum_{t = 0}^{\infty} \mu_t \mathbb{P}(\ccalE_t) \geq \sum_{t = 0}^{\infty} \mu_t - \mu^\prime \delta
		\Rightarrow \mathbb{P}\left( \bigcap_{t = 0}^{T} \ccalE_t \right) \geq 1 - \delta
		\text{.}
\end{equation}
\end{lemma}

\begin{proof}
See Appendix \ref{proof_lemma_safety}.
\end{proof}
Using Lemma~\ref{T:prob_inequality}, we can obtain slacks~$c_i$ such that all feasible policies of~\eqref{P:parametric} are~$(1-\delta_i)$-safe. We state first the result with $\gamma =1$ and finite horizon. 
%
%\begin{theorem}\label{thm_epsilon_finite}
%\blue{If the problem~\eqref{P:parametric} with~$\gamma=1$, finite horizon~$T$ and~$c_i = (T+1)\left(1-\delta_i/(T+1)\right)$ is feasible, then its solution is~$(1-\delta_i)$-safe for the sets~$\ccalS_i$ in the sense of Definition \ref{def_safe_policy}.
%
%Suppose that there exist parameters~$\tilde{\theta}$ such that the policy~$\pi_{\tilde{\theta}}$ is~$(1-\delta_i/(T+1))$-safe for the sets~$\ccalS_i$, with~$i = 1,\ldots,m$. Then the feasibility of the previous problem is guaranteed. }
%
%\end{theorem}
\begin{theorem}\label{thm_epsilon_finite}
{For the problem~\eqref{P:parametric} with~$\gamma=1$, finite horizon~$T$ and~$c_i = (T+1)\left(1-\delta_i/(T+1)\right)$ it follows that
\begin{enumerate}[(i)]
	\item If it is feasible, then its solution is~$(1-\delta_i)$-safe for the sets~$\ccalS_i$ in the sense of Definition \ref{def_safe_policy}.
	\item Further suppose that there exist parameters~$\tilde{\theta}$ such that the policy~$\pi_{\tilde{\theta}}$ is~$(1-\delta_i/(T+1))$-safe for the sets~$\ccalS_i$, with~$i = 1,\ldots,m$, i.e., \eqref{eqn_optimization_problem} is feasible with slack $1-\delta_i/(T+1)$. Then the feasibility of \eqref{P:parametric} is guaranteed. 
\end{enumerate}
}
\end{theorem}
\begin{proof}
Let $\theta^\dagger$ be a solution to \eqref{P:parametric} with $\gamma=1$, finite $T$ and $c_i = (T+1)(1-\delta_i(T+1))$, and denote by $\pi_{\theta^\dagger}$ the policy corresponding to parameters $\theta^\dagger$. Further define the event~$\ccalE_t = \{s_t \in \ccalS_i\}$ and take~$\mu_t=\mu^\prime = 1/(T+1)$ for all~$t = 0,\dots,T$ and $\mu_t=0$ for $t>T$. Since $\theta^\dagger  $ is a solution to \eqref{P:parametric} it follows that 
\begin{equation}
\sum_{t=0}^\infty\mu_t\mathbb{P}(\ccalE_t) \geq \sum_{t=0}^\infty \mu_t -\mu^\prime\delta_i. 
\end{equation}
Thus Lemma \ref{T:prob_inequality} guarantees that 
\begin{equation}
  \mathbb{P}\left( \bigcap_{t = 0}^{T} \{s_t \in \ccalS_i\} \mid \pi_{\theta^\dagger}\right) \geq 1 - \delta_i
		\text{.}
  \end{equation}
To prove the second claim assume that there exists a~$(1-\delta_i/(T+1))$-safe policy for the sets~$\ccalS_i$. Since $\mathbb{P}\left( \bigcap_{t = 0}^{T} \{s_t \in \ccalS_i\} \mid \pi_{\tilde{\theta}}\right) \leq \min_{t} \mathbb{P}(s_t\in \ccalS_i\mid \pi_{\tilde{\theta}}) $, it follows that 
  \begin{equation}
\mathbb{P}\left( \bigcap_{t = 0}^{T} \{s_t \in \ccalS_i\} \mid \pi_{\tilde{\theta}}\right) \leq \frac{1}{T+1}\sum_{t=0}^T \mathbb{P}(s_t\in\ccalS_i\mid \pi_{\tilde{\theta}}).
    \end{equation}
Notice that a $(1-\delta_i/(T+1))$-safe policy ensures that the left hand side of the above equation is bounded below by $(1-\delta_i/(T+1))$. Thus for said policy it follows that 
\begin{equation}
  \frac{1}{T+1}\sum_{t=0}^T \mathbb{P}(s_t\in\ccalS_i\mid \pi_{\tilde{\theta}})\geq 1-\frac{\delta_i}{T+1}.
\end{equation}
Thus, \eqref{P:parametric} is feasible with~$c_i = (T+1)\bigl(1-\delta_i/(T+1)\bigr)$.
\end{proof}

{Theorem~\ref{thm_epsilon_finite} makes two claims. The first, of practical nature, guarantees that a~$(1-\delta_i)$-safe policy can be computed in the finite horizon case by solving~\eqref{P:parametric} with slack variables~$c_i = (T+1)\left(1-\delta_i/(T+1)\right)$. Note that these values can be directly computed using only the desired safety level $\delta_i$ and the operation horizon $T$. Thus, the result does not add any extra hyperparameters. The second claim is of theoretical nature. Indeed, it provides a sufficient condition to guarantee the feasibility of problem~\eqref{P:parametric}. This condition is based on the existence of a safer policy. Verifying the existence of the latter is as hard as solving the original problem \eqref{eqn_optimization_problem}. Hence, the claim is only useful if by experience we know that such a safe policy exists (despite not being available).}

%Note that any solution of~\eqref{eqn_average_prob} with~$\varepsilon_i \geq \delta_i T/(T+1)$ would be~$(1-\delta_i)$-safe. There is, however, no guarantee that the problem is feasible in this case, i.e., that such a policy exists. In fact, we will see in Section~\ref{sec_zdg} that in order to solve~\eqref{eqn_constrained_finite} and find a safe policy in practice, there must actually exist a~$(1-\delta_i/(T+1)+\nu)$-safe policy for some~$\nu > 0$.
In what follows, we derive an analogous result for the discounted problem, i.e., with $\gamma<1$.
%
%\begin{theorem}\label{thm_epsilon_discounted}
%\blue{If the infinite horizon version of problem~\eqref{P:parametric} with~$c_i = (1-\delta_i \gamma^{T_i}(1-\gamma))/(1-\gamma)$ is feasible, then its solution is~$(1-\delta_i)$-safe for the sets~$\ccalS_i$ in the sense of Definition \ref{def_safe_policy}, until time $T_i$.
%
%Suppose that there exist parameters~$\tilde{\theta}$ such that the policy~$\pi_{\tilde{\theta}}$ is~$(1-\gamma^{T_i}(1-\gamma)\delta_i)$-safe for the sets~$\ccalS_i$, with~$i = 1,\ldots,m$. Then the feasibility of the previous problem is guaranteed. }
%\end{theorem}

\begin{theorem}\label{thm_epsilon_discounted}

{For the problem~\eqref{P:parametric} with~$c_i = (1-\delta_i \gamma^{T_i}(1-\gamma))/(1-\gamma)$ it follows that
\begin{enumerate}[(i)]
\item If it is feasible, then its solution is~$(1-\delta_i)$-safe for the sets~$\ccalS_i$ in the sense of Definition \ref{def_safe_policy}, until time $T_i$.
\item Further suppose that there exist parameters~$\tilde{\theta}$ such that the policy~$\pi_{\tilde{\theta}}$ is~$(1-\gamma^{T_i}(1-\gamma)\delta_i)$-safe for the sets~$\ccalS_i$, with~$i = 1,\ldots,m$, i.e., \eqref{eqn_optimization_problem} is feasible with slack $1-\gamma^{T_i}(1-\gamma)\delta_i$. Then the feasibility of \eqref{P:parametric} is guaranteed. 
\end{enumerate}
}
\end{theorem}
\begin{proof}
  We proceed as in the proof of Theorem~\ref{thm_epsilon_finite}. Denote by $\pi_{\theta^\dagger}$ a feasible policy for the problem~\eqref{P:parametric} with~$c_i = (1-\delta_i[1-\gamma^{T_i}(1-\gamma)])/(1-\gamma)$. From the constraint in \eqref{P:parametric} and applying Lemma~\ref{T:prob_inequality} with let~$\ccalE_t = \{s_t \in \ccalS_i\}$, $\mu_t = \gamma^t$ for~$t \geq 0$, and~$\mu^\prime = \gamma^{T_i}$ one has the following implication
\begin{multline}
	U_{i}(\theta^\dagger) = \sum_{t = 0}^{\infty} \gamma^t \mathbb{P}(s_t \in \ccalS_i\mid \pi_{\theta^\dagger})
		\geq \frac{1}{1-\gamma} - \gamma^{T_i}\delta_i
	\\
		\Rightarrow \mathbb{P}\left( \bigcap_{t = 0}^{T_i} \{s_t \in \ccalS_i\} \mid \pi_{{\theta}^\dagger}\right) \geq 1 - \delta_i
		\text{.}
\end{multline}

To establish the second claim, let~$\pi_{\tilde{\theta}}$ be~$(1-\gamma^{T_i}(1-\gamma)\delta_i)$-safe for the sets~$\ccalS_i$, then that policy is such that
  \begin{align}
\min_t \mathbb{P}(s_t \in \ccalS_i\mid \pi_{\tilde{\theta}}) &\geq \mathbb{P}\left( \bigcap_{t = 0}^{\infty} \{s_t \in \ccalS_i\} \mid \pi_{\tilde{\theta}}\right) \nonumber\\
&\geq    1 - \gamma^{T_i}(1-\gamma)\delta_i.
  \end{align}
  Hence it follows that,
\begin{align}
	U_{i}(\tilde{\theta}) = \sum_{t = 0}^\infty \gamma^t \mathbb{P}(s_t \in \ccalS_i\mid \pi_{\tilde{\theta}}) \nonumber\\
		\geq \bigl(1 - \gamma^{T_i}(1-\gamma)\delta_i\bigr) \sum_{t = 0}^\infty \gamma^t
	{}&= \frac{1 - \gamma^{T_i}(1-\gamma)\delta_i}{1-\gamma}
		\text{.}
\end{align}
This shows that the problem is feasible, therefore completing the proof of the result.
\end{proof}

Contrary to the finite horizon problem in Theorem~\ref{thm_epsilon_finite}, Theorem~\ref{thm_epsilon_discounted} does not guarantee that policies obtained from~\eqref{P:parametric} are safe for all~$t \geq 0$, but that they are safe for an arbitrarily long window~$t \leq T_0=\min\{T_1,\ldots,T_m\}$, where $T_1,\ldots,T_m$ are the horizons assumed to exist in Theorem~\ref{thm_epsilon_discounted}. Naturally, arbitrarily safe policies must exist for this to hold and Theorem~\ref{thm_epsilon_discounted} quantifies the trade-off between safety level and safety horizon in terms of the discount factor~$\gamma$. {Notice that this limitation is intrinsic to the infinite horizon problem and not necessarily to the approach here presented. Without further assumptions on the system, to guarantee the probabilistic positive invariance (Definition \ref{def_safe_policy}) one needs to guarantee that infinitely often we have $\mathbb{P}(s_t\in\ccalS_0)=1$. This condition is in general hard to satisfy. Consider for instance the following MDP: a linear systems with Gaussian noise. Notice that even for stable controllers, the probability of being arbitrarily far from the origin is not zero due to the noise. Hence, the condition that $\mathbb{P}(S_t\in\ccalS_0)=1 $ cannot be satisfied for any $t$.} 

Theorems~\ref{thm_epsilon_finite} and~\ref{thm_epsilon_discounted} establish that the safety level of the policies can be preserved as long as safe enough policies exist. Moreover, the levels of safety required and the parameters that define problem \eqref{P:parametric} can be computed based on the desired safety level $\delta_i$ and the desired safety horizon $T_i$. In that sense, we have not lost much by relaxing the chance constraints to be satisfied on~(weighted) average.

{It is fair to point out that for some MDPs the relaxation may result in an optimality gap between problems \eqref{eqn_optimization_problem} and \eqref{P:parametric}.  Indeed, it depends on how the feasible set is reduced. However, without further assumptions it may be impossible to establish a bound. For instance, if the MDP is such that the state transitions are not correlated then there is no modification in the feasible set and thus the solutions are the same. If one has more information about the MDP it may be possible to bound the difference.  This analysis is beyond the scope of this work since we are mainly concerned with safe operation.}

At this point we have all the elements to formulate the problem \eqref{P:parametric} so that the solutions of the problem are guaranteed to satisfy the desired safety requirements. Thus, we turn our attention to solving this problem. In Section \ref{sec_safe_learning} we briefly discussed the benefits of solving the dual relaxation \eqref{P:dual_parametric} from the theoretical and practical point of view. In Section \ref{sec_pd} we dive deeper in the algorithmic details, while in the next section we focus on the theoretical guarantees of solving the dual relaxation.

Recall that the dual problem provides the tightest upper bound on the problem of interest \eqref{P:parametric} regardless of the structure of the problem (weak duality). However, given its non-convexity we do not have guarantees on how big this gap is a priori. The next section establishes the almost optimality of the dual relaxation \eqref{P:dual_parametric} by showing that the duality gap can be made arbitrarily small. It is worth pointing out that despite these results being developed in the context of RL, they hold for any stochastic control problem of the form \eqref{P:parametric}.

%!TEX root = root.tex

\section{Safe Learning Has (Almost) Zero Duality Gap}
\label{sec_zdg}

Having established that the relaxation proposed allows a relatively easy computation of the dual problem associated with \eqref{eqn_optimization_problem}, we focus on understanding the advantages of working in the dual domain. First, notice that the dual function (cf., \eqref{eqn_parametric_dual_function}) is a point-wise maximum of linear functions, and therefore it is convex \cite[Section 3.2.3]{boyd2004convex}. Hence, solving the problem in the dual domain is simple since gradient descent on the multipliers finds the optimal dual variable. However, in general, the dual problem only provides an upper bound on the solution of~\eqref{P:parametric}. The difference between the dual optimum~\eqref{P:dual_parametric} and the solutions of~\eqref{P:parametric} is termed the duality gap. Despite~\eqref{P:parametric} being non-convex we show in Section~\ref{sec_params} that the duality gap is arbitrarily small for a sufficiently rich parametrization. This implies that, for the optimal value of $\lambda$, maximizing~\eqref{eqn_parametric_lagrangian} is equivalent to solving~\eqref{P:parametric}. %In Section~\ref{sec_pd} this result is exploited to establish the convergence of a {primal-dual} algorithm to learn a $(1-\delta)$-safe policy.

In the rest of this section, we focus in showing that the duality gap can be made arbitrarily small for rich enough parameterization (Section \ref{sec_params}), before doing so, we require an auxiliary result regarding the case where policies are stochastic kernels. %Instead, they are arbitrary probabilities in the space of probability measures over the state-action space. 
We show that in this case, the safe-learning problem has zero duality gap. This is the subject of Section \ref{sec_non_parametric}. 
\subsection{The case of stochastic kernels}\label{sec_non_parametric}
To be precise, we consider in this section that the agent chooses actions sequentially based on a policy~$\pi \in \ccalP_s(\ccalA)$, where $\ccalP_s(\ccalA)$ is the space of stochastic kernels with source the state space $\ccalS$ and target $\ccalA$. 
%probability measures on~$(\ccalA, \ccalB(\ccalA))$ parametrized by elements of~$\ccalS$, where~$\ccalB(\ccalA)$ are the Borel sets of~$\ccalA$.
Given the relation between the continuous task and the episodic problems discussed in Remark \ref{remark_equivalence}, we will develop the results for the discounted problem. In this case, the safe-learning problem yields
\begin{equation}\label{P:constrainedRL}
  \begin{aligned}
	P^\star \triangleq \max_{\pi \in \ccalP_s(\ccalA)}& \hspace{-2ex}&V(\pi) &\triangleq
		\mathbb{E} \left[ \sum_{t = 0}^\infty \gamma^t r(s_t,a_t) \mid \pi\right]
	\\
	\subjectto \hspace{1ex} & \hspace{-2ex}&U_i(\pi) &\triangleq
		\mathbb{E} \left[ \sum_{t = 0}^\infty \gamma^t \mathbbm{1}(s_t\in\ccalS_i) \mid \pi\right] \geq c_i%,\, i=1,\ldots, m
		\text{.}
  \end{aligned}
  \end{equation}
%
%where $c_i = 1+\varepsilon/(1-\delta)$.
To be formal let us define the dual problem associated to \eqref{P:constrainedRL}. To do so, let $\lambda\in\mathbb{R}_+^m$ and as done before we define the Lagrangian 
\begin{equation}\label{E:lagrangian}
	\ccalL(\pi,{\lambda}) \triangleq V(\pi) + \sum_{i=1}^m\lambda_i \left(U_i(\pi)-c_i\right)
		\text{.}
\end{equation}
%
%The dual function is then the point-wise maximum of~\eqref{E:lagrangian} with respect to the policy~$\pi$, i.e.,
%
%\begin{equation}\label{E:dual_function}
%	d({\lambda}) \triangleq \max_{\pi\in\ccalP(\ccalS)} \ccalL(\pi,{\lambda}),
%		\text{.}
%\end{equation}
%
Then, analogous to the parametric case, we define the dual problem as
%The dual function~\eqref{E:dual_function} provides an upper bound on the value of~\eqref{P:constrainedRL}, i.e., $d({\lambda}) \geq P^\star$ for all~${\lambda} \in\mathbb{R}^m_+$~\cite[Section 5.1.3]{boyd2004convex}. The tighter the bound, the closer the policy obtained from~\eqref{E:dual_function} is to the optimal solution of~\eqref{P:constrainedRL}. Hence, the dual problem is that of finding the tightest of these bounds:
%
\begin{equation}\label{P:dualRL}
  \begin{aligned}
	D^\star \triangleq \min_{{\lambda} \in \mathbb{R}^m_+}  \max_{\pi\in\ccalP(\ccalS)} \ccalL(\pi,{\lambda}).
  \end{aligned}
  \end{equation}
Note that the dual problem can be related to the unconstrained, risk aware problem~\eqref{eqn_risk_aware_reward} by taking~$\lambda_i = w_i$ in~\eqref{E:lagrangian}. %Hence, \eqref{E:dual_function} takes on the optimal value of~\eqref{P:regularizedRL} for all possible regularization parameters.
Problem~\eqref{P:dualRL} then finds the best regularized problem, i.e., that whose value is closest to~$P^\star$. %It turns out, this problem is tractable if~$d({\lambda})$ can be evaluated, since~\eqref{P:dualRL} is a convex program~(the dual function is the point-wise maximum of a set of linear functions and is therefore convex)~\cite[Section 3.2.3]{boyd2004convex}.
{Despite these similarities,~\eqref{P:dualRL} does not necessarily solve the same problem as~\eqref{P:constrainedRL}. In other words, there need not be a relation between the optimal dual variables~${\lambda}^\star$ from~\eqref{P:dualRL} or the regularization parameters~$w_i$ in \eqref{eqn_risk_aware_reward} and the constraints~$c_i$ of~\eqref{P:constrainedRL}. This depends on the value of the duality gap~$\Delta = D^\star-P^\star$. Indeed, if~$\Delta$ is small, then so is the suboptimality %and infeasibility
  of the policies obtained from~\eqref{P:dualRL}. In the limit case where~$\Delta = 0$, %\blue{MCF:Explain}\red
  {problems~\eqref{P:constrainedRL} and~\eqref{P:dualRL} %and~\eqref{P:regularizedRL} would all be essentially equivalent.}
    would be equivalent. Since~\eqref{P:constrainedRL} is not a convex program, however, this result does not hold immediately. The subject of Theorem \ref{T:strongDuality} is to establish that in fact~\eqref{P:constrainedRL} has zero duality gap under mild conditions.%basically the same mild conditions as convex programs.
  }
%
%\begin{equation}\label{P:constrainedRL}
%  \begin{aligned}
%	P^\star \triangleq \max_{\pi \in \ccalP_s(\ccalA)}& \hspace{-2ex}&V(\pi) &\triangleq
%		\mathbb{E} \left[ \sum_{t = 0}^\infty \gamma^t r(s_t,a_t) \mid \pi\right]
%	\\
%	\subjectto \hspace{1ex} & \hspace{-2ex}&U_i(\pi) &\triangleq
%		\mathbb{E} \left[ \sum_{t = 0}^\infty \gamma^t \mathbbm{1}(s_t\in\ccalS_i) \mid \pi\right] \geq c_i%,\, i=1,\ldots, m
%		\text{.}
%  \end{aligned}
%  \end{equation}
%
{This result has been established in \cite{cPaternainEtal2019d} based on an analysis of the perturbation function associated to problem \eqref{P:constrainedRL}, defined as   
	\begin{equation}\label{P:perturbed}
	\begin{aligned}
	P(\xi) \triangleq&\max_{\pi \in \ccalP_s(\ccalA)} & \hspace{-2ex} &V(\pi)
	\\
	& \hspace{1.5ex}\subjectto & \hspace{-2ex} &U_i(\pi) \geq c_i+\xi_i \quad \forall i=1,\ldots, m.
	\end{aligned}\end{equation}
	In this work we establish an alternative geometric proof which consists on showing that the sub-level set of the vector valued function $[V(\pi),U_1(\pi),\ldots,U_m(\pi)]$ is convex. This is the subject of the following lemma. 
	\begin{lemma}\label{L:convexity}
	Suppose that~$r(s,a)$ is bounded for all~$(s,a)\in\ccalS\times\ccalA$ and problem \eqref{P:constrainedRL} is feasible. Then, the following set is non-empty and convex. 
		\begin{align}\label{eqn_convex_set}
		\ccalC = \bigl\{\xi \in \mathbb{R}^{m+1}\big| & \exists\, \pi \in \ccalP_{s}(\ccalA), V(\pi) \geq \xi_0, \nonumber\\
		&U_i(\pi) \geq c_i+\xi_i, \forall i=1,\ldots, m \bigr\}.
		\end{align} 
	\end{lemma}
	\begin{proof}
		See Appendix \ref{proof_convexity}.
	\end{proof}
	Now we are in conditions of stating the duality result. }
%
%Before stating the result we define the perturbation function which is fundamental in the proof of the result and it is also required for further developments. For any $\xi\in\mathbb{R}^m$, the perturbation function is defined as 
  %
%We next formalize the zero duality gap of problem~\eqref{P:constrainedRL}.
%
\begin{theorem}\label{T:strongDuality}
  Suppose that~$r(s,a)$ is bounded for all~$(s,a)\in\ccalS\times\ccalA$ and that there exists a strictly feasible policy, i.e., there exists $\pi^\ddagger \in \ccalP_s(\ccalA)$ such that $U_i(\pi^\ddagger)> c_i$ for all $i=1,\ldots,m$.  Then, strong duality holds for~\eqref{P:constrainedRL}, i.e., $P^\star = D^\star$.
\end{theorem}
\begin{proof}
	
		Weak duality establishes that $D^\star \geq P^\star$~\cite[Section 5.1.3]{boyd2004convex}. Hence we are left to establish that $P^\star \geq D^\star$. To do so, it suffices to show that there exists $\lambda^\dagger\in\mathbb{R}^m_+$ such that 
\begin{equation}
P^\star\geq d(\lambda^\dagger) =\max_{\pi\in\ccalP(\ccalS)} \ccalL(\pi,{\lambda}^\dagger).
\end{equation}

		Denote by $\boldsymbol{0}_m\in\mathbb{R}^m$ the null vector and notice that $(P^\star,\boldsymbol{0}_m^\top)^\top \in \ccalC$, where $\ccalC$ is the set defined in \eqref{eqn_convex_set}. In addition, we have that $(P^\star,\boldsymbol{0}_m^\top)^\top \notin \interior(\ccalC)$. Indeed, if it were, there would exist a ball centered in $(P^\star,\boldsymbol{0}_m^\top)^\top$ and contained in $\ccalC$, which in turn would contradict the optimality of $P^\star$. 
		
		Since $\ccalC$ is a convex set (cf., Lemma \ref{L:convexity}) and $(P^\star,\boldsymbol{0}_m^\top)^\top$ is in its boundary, the supporting hyperplane theorem~\cite[Section 2.5.2]{boyd2004convex} establishes that there exists $\tilde{\lambda} \in \mathbb{R}^{m+1}$ such that 
		\begin{equation}\label{eqn_separating}
		(P^\star,\boldsymbol{0}_m^\top) \tilde{\lambda} \geq \xi^\top \tilde{\lambda},
		\end{equation}
		for all $\xi \in \ccalC$.  We claim that $\tilde{\lambda}$ belongs to the nonnegative orthant and that its first component $\tilde{\lambda}_0$ is strictly positive. We prove these facts at the end of this proof. Let $\lambda^\dagger = \tilde{\lambda}/\tilde{\lambda}_0$ and 
		\begin{equation}
		\pi^\dagger = \argmax_{\pi\in\ccalP_{s}(\ccalA)} V(\pi) + \sum_{i=1}^m \lambda^\dagger_i (U_i(\pi)-c_i).
		\end{equation}
	Then, let $\xi^\dagger = (V(\pi^\dagger), U(\pi^\dagger)-c)$ and notice that $\xi^\dagger\in \ccalC$. Using this definition, write the dual function associated to the problem \eqref{P:constrainedRL} evaluated at $\lambda^\dagger$ as  
		\begin{equation}
d(\lambda^\dagger) = \ccalL(\pi^\dagger,\lambda^\dagger) = (1,\lambda^\dagger)^\top \xi^\dagger =\left(\tilde{\lambda}/\tilde{\lambda}_0\right)^\top \xi^\dagger \leq P^\star,
		\end{equation} 
		where the last inequality follows from \eqref{eqn_separating}. This completes the proof that there exists $\lambda^\dagger$ such that $P^\star \geq d(\lambda^\dagger)$. Moreover, $\lambda^\dagger$ belongs to the set of optimal dual multipliers. 
		
		To complete the proof, we are left to show that indeed $\tilde{\lambda}\in\mathbb{R}_+^{m+1}$ and that its first component is strictly positive. We argue by contradiction. Assume that for some $i=0,1,\ldots,m$, we have that $\tilde{\lambda_i}<0$. Since $\ccalC$ is not bounded below, i.e., we can choose $\xi_i$ arbitrarily small. Hence, there exists $\xi\in\ccalC$ such that \eqref{eqn_separating} is violated. This proves that $\tilde{\lambda}\in\mathbb{R}^{m+1}_+$.  

Likewise, assume that $\tilde{\lambda_0} = 0$. Then, \eqref{eqn_separating} reduces to 
		\begin{equation}\label{eqn_separating_aux2}  
		0 \geq \sum_{i=1}^m\xi_i\tilde{\lambda}_i. 
		\end{equation}
		Recall that $\tilde{\lambda}\in\mathbb{R}^{m+1}_+$ and by virtue of the supporting hyperplane theorem $\tilde{\lambda} \neq 0$. Hence, for  \eqref{eqn_separating_aux2} to hold for every $\xi \in\ccalC$, at least one index $i=1,\ldots, m$ is such that $\xi_i<0$ or that $\xi=0$.  This contradicts the fact that the problem \eqref{P:constrainedRL} is strictly feasible. Indeed, notice that since the problem is strictly feasible there is $\xi\in\ccalC$ such that $\xi_i>0$ for $i=1,\ldots,m$.   

\end{proof}
Theorem~\ref{T:strongDuality} establishes a fundamental equivalence between the primal~\eqref{P:constrainedRL} and the dual~\eqref{P:dualRL} problems. Indeed, since~\eqref{P:constrainedRL} has no duality gap, its solution can be obtained through~\eqref{P:dualRL}.{ What is more, the trade-offs expressed by~$w_i$ in~\eqref{eqn_risk_aware_reward} are the same as those expressed by the constraint specification~$c_i$ in the sense that they trace the same Pareto front. Nevertheless, the relationship between~$c_i$ and~$w_i$ is not trivial and that specifying the constrained problem is often considerably simpler. 

%As an example of this, the trust region policy optimization algorithm~(TRPO) is obtained by transforming a regularized problem into a constrained one~\cite{schulman2015trust}. Theorem~\ref{T:strongDuality} establishes that this is indeed a valid transformation, since both problems are equivalent. Observe that due to the non-convexity of the objective in RL problems, this result is in fact not immediate.}

The theoretical importance of the previous result notwithstanding, it does not yield a procedure to solve~\eqref{P:constrainedRL} since evaluating the dual function involves a maximization problem that is intractable for general classes of distributions. In the next section, we study the effect of using a finite parametrization for the policies and show that the price to pay in terms of duality gap depends on how ``good'' the parametrization is. If we consider, for instance, a neural network---which are universal function approximators~\cite{funahashi1989approximate, cybenko1989approximation, hornik1989multilayer, lu2017expressive, lin2018resnet}---the loss in optimality can be made arbitrarily small. %We then proceed to derive an explicit algorithm to solve the constrained RL problem~\eqref{P:constrainedRL}.
%
%%%%%%%%%%%%%%%%%%%%%%%%%%%%%%%%%%%%%%%%%%%%%%%%%%%%%%%%%%%5
%%%%%%%%% S E C T I O N %%%%%%%%%%%%%%%%%%%%%%%%%%%%%%%%%%%%
%%%%%%%%%%%%%%%%%%%%%%%%%%%%%%%%%%%%%%%%%%%%%%%%%%%%%%%%%%%
%
\subsection{Back to parametric policies}
\label{sec_params}
Having established the zero duality gap in the case of stochastic kernels we turn our attention back to the case of policies parameterized by $\theta \in \mathbb{R}^d$. In this section, we focus our attention on a widely used class of parameterizations that we term \textit{near-universal}, which are able to model any stochastic kernel in $\ccalP_s(\ccalA)$ to within a stated accuracy. We formalize this concept in the following definition.
\begin{definition}\label{def_universal_param}
  Let $\epsilon>0$, a parametrization $\pi_\theta$ is an $\epsilon$-universal parametrization of functions in $\ccalP_s(\ccalA)$ if for any $\pi\in\ccalP_s(\ccalA)$ there exists a parameter $\theta \in \mathbb{R}^d$ such that
  \begin{equation}
\max_{s\in\ccalS}\int_{\ccalA}\left|\pi(a|s)-\pi_\theta(a|s)\right|\,da\leq \epsilon.
  \end{equation}  
  \end{definition}
The previous definition includes all parameterizations that induce distributions that are close to distributions in $\ccalP_s(\ccalA)$ in total variational norm. Notice that this is a milder requirement than approximation in uniform norm which is a property that has been established to be satisfied by radial basis function networks%(RBFNs)
~\cite{park1991universal}, reproducing kernel Hilbert spaces%(RKHS)
~\cite{sriperumbudur2010relation} and deep neural networks% (DNNs)
~\cite{hornik1989multilayer}. {Since the objective function and the constraints in Problem \eqref{P:constrainedRL} involve an infinite horizon, the policy is applied an infinite number of times. Hence, the error introduced by the parametrization could accumulate and induce distributions over trajectories that differ considerably. We claim in the following lemma that this is not the case. 

To make the statement precise we are required to define the occupation measure of a discounted MDP. { Let $p_\pi^t(s,a)$ be the Radon-Nikodym derivative of the measure $\mathbb{P}(s_{t}\in \tilde{\ccalS},a_{t} \in \tilde{\ccalA} \mid \pi) $.} The occupancy measure is then a discounted average of $p_\pi^t(s,a)$ (see e.g. \cite{borkar1988convex})
  \begin{equation}\label{eqn_occupation_measure}
    \rho(s,a) = \left(1-\gamma\right)\sum_{t=0}^\infty \gamma^tp_\pi^t(s,a).
    \end{equation}

We are now in conditions of stating the following result.
\begin{lemma}\label{lemma_meassure_bound}
  Let $\rho$ and $\rho_\theta$ be occupation measures induced by the policies $\pi\in\ccalP_s(\ccalA)$ and $\pi_{\theta}$ respectively, where $\pi_{\theta}$ is an $\epsilon$- parametrization of $\pi$. Then, it follows that
  \begin{equation}\label{eqn_occupation_bound}
\int_{\ccalS\times\ccalA} |\rho(s,a)-\rho_\theta(s,a)| \,dsda \leq \frac{\epsilon}{1-\gamma}.
    \end{equation}
\end{lemma}
\begin{proof}
See Appendix \ref{proof_meassure_bound}.
  \end{proof}
The previous result, although derived as a technical result required to bound the duality gap for parametric problems, has a natural interpretation. The smaller $\gamma$\textemdash the less we are concerned about rewards far in the future\textemdash the smaller the error in the approximation of the occupation measure.

%Having defined the concept of universal approximator, we formalize the parametric safe learning problem
%
%\begin{equation}\label{P:parametric}
%  \begin{aligned}
%	P^\star_\theta\triangleq\max_{\theta\in\ccalH}& &&V(\theta)\triangleq	\mathbb{E} \left[ \sum_{t = 0}^\infty \gamma^t r(s_t,a_t) \mid \pi_\theta \right]
%	\\
%	\subjectto& &&U_i(\theta)  \triangleq	\mathbb{E} \left[ \sum_{t = 0}^\infty \gamma^t \mathbbm{1}\left(s_t\in\ccalS_i\right) \mid \pi_\theta\right]\geq c_i.
%  \end{aligned}
%  \end{equation}
%
%Notice that the problem~\eqref{P:parametric} is equivalent to~\eqref{eqn_constrained_discounted}, thus the dual function associated to the problem is the one defined on \eqref{eqn_parametric_dual_function},
%
%As done in the previous section, let~$\lambda\in\mathbb{R}_+^m$  and define the dual function associated to~\eqref{P:parametric} as
%
%\begin{equation}\label{eqn_parametric_dual_function}
%d_\theta(\lambda) \triangleq \max_{\theta\in\ccalH}\ccalL_\theta(\theta,\lambda) \triangleq \max_{\theta\in\ccalH} V(\theta) + \lambda\left(U(\theta)-c\right), 
%  \end{equation}
%
%Likewise we define the dual problem as finding the tightest upper bound for \eqref{P:parametric}
%
%\begin{equation}\label{P:dual_parametric}
%  \begin{aligned}
%	D_\theta^\star \triangleq\min_{{\lambda} \in \mathbb{R}^m_+}  d_\theta(\lambda). 
%  \end{aligned}
%  \end{equation}
%

When the policies were selected from arbitrary distributions we showed in Theorem \ref{T:strongDuality} that the safe learning problem has zero duality gap and thus, $P^\star=D^\star$. This is no longer the case when we consider parametric policies. However, we claim and prove in the following theorem that the duality gap is bounded by a function that is linear with the approximation error $\epsilon$ of the parameterization. 
%As previously stated, the reason for introducing the parametrization is to turn the original functional optimization problem into a tractable problem in which the optimization variable is a finite dimensional vector of parameters. Yet, there is a cost for introducing the aforementioned parametrization: the duality gap is no longer null. The latter means that the solution obtained through the dual problem is sub-optimal. We claim however that this gap is bounded by a function that is linear with the approximation error $\epsilon$, and thus if the parametrization has a good representation power the price to pay is almost zero. This is the subject of the following theorem. 
%
\begin{theorem}\label{T:parametrization}
  Suppose that~$r(s,a)$ is bounded for all~$(s,a)\in\ccalS\times\ccalA$ by a constant $B_r>0$ and that there exists a strictly feasible solution for \eqref{P:constrainedRL}. Let $\lambda_\epsilon^\star$ be the solution to the dual problem associated with the perturbed problem \eqref{P:perturbed} with perturbation $\xi_i=\epsilon/(1-\gamma)$ for all $i=1,\ldots,m$. Then, if the parametrization $\pi_\theta$ is an $\epsilon-$universal parametrization of functions $\pi\in\ccalP_s(\ccalA)$ it follows that 
  \begin{equation}
 \left(B_{r}+\left\|\lambda_\epsilon^\star\right\|_1\right)\frac{\epsilon}{1-\gamma}+    P_\theta^\star\geq D_\theta^\star \geq P_\theta^\star
  \end{equation}
  where~$P^\star_\theta$ is the optimal value of~\eqref{P:parametric} and~$D_{\theta}^\star$ the value of the parametrized dual problem \eqref{P:dual_parametric}.
  \end{theorem}
\begin{proof}
See Appendix \ref{proof_param}.
  \end{proof}
The implication of the previous result is that there is almost no price to pay by introducing a parametrization. By solving the dual problem \eqref{P:dual_parametric} the sub-optimality achieved is of order $\epsilon$, i.e., the error on the representation of the policies. Notice that this error could be made arbitrarily small by increasing the representation ability of the parametrization, by for instance increasing the dimension of the vector of parameters $\theta$. 

{The latter means that the main burden of solving \eqref{P:parametric} is in computing the dual function. Which is equivalent to solving an unconstrained RL problem (cf., \eqref{reward_constrained} and \eqref{eqn_parametric_lagrangian} and the related discussion). Thus, solving constrained RL problems of the form \eqref{P:parametric} is not much more challenging than its unconstrained counterpart. Moreover, working on the dual domain provides the advantages discussed in Section \ref{sec_safe_learning} as compared to working the risk aware reward.  

Based on the previous results we can exploit dual and primal-dual algorithms to solve the safe RL problem \eqref{P:parametric}. This is the subject of the next section.

%\input{gradient}
%\input{pd}
%!TEX root = root.tex

\section{Primal-Dual Algorithm}\label{sec_pd}
As previously stated, the advantage of working in the dual domain is that the dual function~\eqref{eqn_parametric_dual_function} is convex and therefore it can be minimized using gradient descent (or its stochastic variants). Primal-dual algorithms are well-known and they have been considered extensively in the optimization literature \cite{arrow_hurwicz,nedic2009subgradient,paternain2016online} and specifically for RL \cite{chow2017risk}. We include the details of these algorithms next since we use them to solve the problem \eqref{P:parametric}. It is worth pointing out that even though it is known that primal-dual algorithms solve the dual problem they do not provide any information regarding the solution of the primal. By establishing a bound in the duality gap (cf., Theorem \ref{T:strongDuality} and Theorem \ref{T:parametrization}) we can claim that solving the dual problem is a good approximation of the primal. 

For completeness we start by considering plain dual gradient descent. Let $\nabla_\lambda d_\theta(\lambda)$ be the gradient of the dual function, then dual gradient descent is defined by the following update rule
\begin{equation}\label{eqn_pure_gradient_descent}
\lambda^{k+1} = \left[\lambda^k-\eta_\lambda \nabla_\lambda d_\theta(\lambda^k)\right]^+,
  \end{equation}
where $\eta_\lambda>0$ is a step-size and the operation $\left[\cdot\right]^+$ denotes a projection onto the positive orthant of $\mathbb{R}^m$. The latter algorithm is ensured to converge to a neighborhood of the dual optimum \cite[pp 43-45]{bertsekas1999nonlinear}. In addition, the gradient of the dual function can be computed using Dankin's Theorem (see e.g. \cite[Chapter 3]{bertsekas15c}) by evaluating the constraints in the original problem \eqref{P:parametric} at the primal maximizer of the Lagrangian, i.e.,
\begin{equation}\label{eqn_dual_gradient}
\nabla_\lambda d_\theta(\lambda^k) = U(\theta^\star(\lambda^k)) -c_i,
\end{equation}
where $\theta^\star(\lambda^k) \triangleq \argmax_{\theta\in\mathbb{R}^d}\ccalL(\lambda^k,\theta)$. As discussed in Section \ref{sec_safe_learning} computing the maximizer of the Lagrangian for a given $\lambda$ is equivalent to maximizing the expected discounted control cost for the risk aware reward~\eqref{eqn_risk_aware_reward} and~\eqref{reward_constrained}.

Depending on the framework of interest this maximization can be done using several techniques. In the case of model-free scenarios, classic RL algorithms such as policy gradient~\cite{williams1992simple,sutton2000policy,paternain2018stochastic} or actor-critic methods~\cite{konda2000actor} can be used. Since these algorithms are gradient based, they are not guaranteed to converge to a global maximum {(unless the state-action space is finite \cite{agarwal2020optimality,ding2020natural}}). {However, their empirical success suggests that they converge to solutions with small suboptimality.} This is, we can compute the iteration $k+1$ of the policy's parameter as 
\begin{equation}\label{eqn_primal_max}
\ccalL_\theta(\theta^{k+1},\lambda^k) \approx \max_{\theta\in\mathbb{R}^d} \ccalL_\theta(\theta,\lambda^k).
\end{equation}

{Then, the dual variable is updated following the gradient descent scheme suggested in~\eqref{eqn_pure_gradient_descent}, where we replace the gradient of the dual function given in \eqref{eqn_dual_gradient} by the approximation based on the primal variable available $\theta^{k+1}$. This yields the following update 
\begin{equation}\label{eqn_gradient_descent}
\lambda^{k+1} =\left[\lambda^k - \eta_\lambda \left(U(\theta^{k+1})-c\right)\right]_+.
  \end{equation}
The algorithm given by~\eqref{eqn_primal_max}--\eqref{eqn_gradient_descent} is summarized for convenience under Algorithm~\ref{alg_dual_descent}.}
%
%%%%%%%%%%%%%%%%%%%%%%%%%%%%%%%%%%%%%%%%%%%%%%%%%%%%%%%%%%%%%%%%%%%%%%%%%%%%%%%%%%%%%%%%%%%%%%%%%%%%%%%%%%%%%%%%%%%%%%%%%%% A L G O R I T H M %%%%%%%%%%%%%%%%%%%%%%%%%%%%%%%%%%%%%%%%%%%%%%%%%%%%%%%%%%%%%%%%%%%%%%%%%%%%%%%%%%%%%%%%%%%%%%%%%%%%%%%%%%%%%%%%%%%%%%%%%%%%%%%%%%%%%%%%%%%%%%%%%%%%%%%%%%%%%%%
\begin{algorithm}[t]
  \caption{Dual Descent}
  \label{alg_dual_descent} 
\begin{algorithmic}[1]
 \renewcommand{\algorithmicrequire}{\textbf{Input:}}
 \renewcommand{\algorithmicensure}{\textbf{Output:}}
 \Require $\eta_\lambda$
 \State \textit{Initialize}: $\theta^0 = 0$, $\lambda^0 = 0$
  \For {$k=0,1\ldots$}
  \State Compute a primal approximation via an RL algorithm to get $\theta^{k+1}$ such that 
   \begin{align*}
      \ccalL_{\theta}(\theta^{k+1},\lambda^k) \approx \max_{\theta\in\mathbb{R}^d} \ccalL_{\theta}(\theta,\lambda^k)
  \end{align*}

  \State Compute the dual ascent step 
  \begin{align*}
  \lambda^{k+1} = \left[\lambda^k-\eta_\lambda \left(U(\theta^{k+1})-c\right)\right]_+
  \end{align*}
  \EndFor
 \end{algorithmic}
 \end{algorithm}
The previous expression approximates a dual descent step since the maximization is not exact. However, it can still be guaranteed to converge to a neighborhood of the solution under mild assumptions \cite{fazlyab2018distributed}. 

The bound on the number of iterations required is linear with the inverse of the desired accuracy $\varepsilon$ \cite[pp 43-45]{bertsekas1999nonlinear}. This is the overhead of solving the dual problem as compared to solving the risk aware optimal control problem for a hyperparameter $\lambda$. The small overhead (in proportion to solving the unconstrained RL problem) makes this approach beneficial as compared to that of selecting a hyperparameter in ad-hoc forms as we have discussed in Section \ref{sec_safe_learning}. To complete this discussion, we have in fact introduced two new parameters, the initial condition $\lambda_0$ and the step-size $\eta_\lambda$. Notice that this substitution is indeed beneficial. The algorithm is guaranteed to converge for any small enough $\eta_\lambda$, which makes this selection easier than that of tuning the hyperparameter $\lambda$. Moreover, since the dual function is convex the algorithm is guaranteed to converge for any initial condition $\lambda_0$. 

Granted, often the cost of running policy gradient or actor-critic algorithms until convergence before updating the dual variable results in an algorithm that is computationally prohibitive. As an alternative, it is common in the context of optimization to update the primal and dual variables in parallel \cite{arrow_hurwicz}. This idea can be applied in the context of RL as well, where a policy gradient \textemdash or actor-critic as in \cite{bhatnagar2012online, tessler2018reward}\textemdash update is followed by an update of the multipliers along the direction of the constraint slack. In these algorithms the update on the policy is on a faster scale than the update of the multipliers, and therefore they operate, from a theoretical point of view, as dual descent algorithms (cf., Algorithm \ref{alg_dual_descent}). In particular, the proofs in \cite{bhatnagar2012online, tessler2018reward} rely on the fact that these timescales are such that it allows to consider the multiplier as constant. Thus, the maximization in \eqref{eqn_primal_max} is replaced by a gradient ascent step
\begin{equation}\label{eqn_first_grad_ascent}
  \theta^{k+1} = \theta^k + \eta_{\theta} \nabla_\theta \ccalL(\theta^k,\lambda^k), 
  \end{equation}
where $\eta_\theta>0$ is the step-size of the primal ascent and as previously discussed $\eta_{\theta}\gg \eta_\lambda$.

Notice that the computation of the expressions \eqref{eqn_first_grad_ascent} and \eqref{eqn_gradient_descent} is not straightforward since they involve expectations across trajectories. However, the gradient with respect to $\theta$ can be done using the Policy Gradient Theorem \cite{sutton2000policy}. Let us define
      %%%%%%%%%%%%%%%%%%%%%%%%%%%%%%%%%%%%%%%%%%%%%%%%%%%%%%%%%%%%%%%%%%%%%%%%%%%%%%%%%%%%%%%%%%%%%%%%%%%%%%%%%%%%%%%%%%%%%%%%%%% A L G O R I T H M %%%%%%%%%%%%%%%%%%%%%%%%%%%%%%%%%%%%%%%%%%%%%%%%%%%%%%%%%%%%%%%%%%%%%%%%%%%%%%%%%%%%%%%%%%%%%%%%%%%%%%%%%%%%%%%%%%%%%%%%%%%%%%%%%%%%%%%%%%%%%%%%%%%%%%%%%%%%%%%
\begin{algorithm}[t]
  \caption{Stochastic Primal-Dual for Safe Policies}
  \label{alg_pd} 
\begin{algorithmic}[1]
 \renewcommand{\algorithmicrequire}{\textbf{Input:}}
 \renewcommand{\algorithmicensure}{\textbf{Output:}}
 \Require $\theta^0,\lambda^0,T, \eta_\theta,\eta_\lambda, \delta, \epsilon$
 %\State \textit{Initialize}: $\hat{Q} = 0$, $s_0 = s$, $a_0 =a$
 % \State Draw an integer $T_Q$ form a geometric distribution with parameter $\gamma$, $P(T_Q = t) = (1-\gamma)\gamma^t$
% \\ \textit{LOOP Process}
  \For {$k = 0,1,\ldots $}
  \State Simulate a trajectory with the policy $\pi_{\theta^k}(\bbs)$
  \State Estimate primal gradient $\hat{\nabla}_\theta\ccalL(\theta^k,\lambda^k)$ as in \eqref{eqn_stochastic_primal}
  \State Estimate dual gradient $\hat{U}(\theta^k)-s$ as in \eqref{eqn_stochastic_dual}
  \State Update primal variable
  $$
  \theta^{k+1} = \theta^k+ \eta_{\theta} \hat{\nabla}_\theta\ccalL(\theta^k,\lambda^k)
    $$
    \State Update dual variable
  $$
  \lambda^{k+1} = \left[\lambda^k+ \eta_{\lambda}\left(\hat{U}(\theta^k)-c\right)\right]_+
    $$
  \EndFor
% \Return $\theta$, $s_{T_Q}$ 
 \end{algorithmic}
 \end{algorithm}
%      
%
%\begin{equation} \label{eqn_cum_reward_finite}
%R_{T}^{\lambda}(\bbs,\bba) = \sum_{t=0}^T r_\lambda(s_t,a_t),
%\end{equation} 
%
\begin{equation} \label{eqn_cum_reward_infinite}
R^{\lambda}(\bbs,\bba) = \sum_{t=0}^\infty \gamma^tr_\lambda(s_t,a_t)
\end{equation}
where $r_\lambda(s_t,a_t)$ is the reward defined in \eqref{reward_constrained} for a given Lagrange multiplier $\lambda\in \mathbb{R}_+$. In addition, let 
  \begin{equation}
    Q^\lambda(s,a) = \mathbb{E}\left[R^\lambda(\bbs,\bba)| s_0=s, a_0 =a\right].
    \end{equation}
Further recall the definition of the occupation measure in the discounted case \eqref{eqn_occupation_measure}. In the case of the undiscounted reward the occupation measure is alternatively defined as $\rho_\theta(s) = \sum_{t=0}^T p^t_{\pi_\theta}(s|s_0)/T$. Then the gradient of the Lagrangian \eqref{eqn_lagrangian} with respect to the parameters of the policy $\theta$ for both formulations yields \cite{sutton2018reinforcement}
    \begin{equation}\label{eqn_primal_gradient}
 \nabla_\theta \ccalL(\theta,\lambda) =  \mathbb{E}_{a\sim\pi_{\theta}(a|s),s\sim \rho_\theta(s)}\left[Q^\lambda(s,a) \nabla_{\theta}\log \pi_{\theta}(a|s) \right].
      \end{equation}
In both expressions of the gradient with respect to the primal and the dual variables we require to compute expectations with respect to the trajectories of the system. To avoid sampling a large number of trajectories, one can instead use stochastic approximations \cite{robbins1951stochastic}. This is, with one sample trajectory, one can compute in the case of the finite horizon problem estimates of $U(\theta^k)$ as
    \begin{equation}\label{eqn_stochastic_dual}
      \hat{U}_i(\theta^k) =\sum_{t=0}^T\mathbbm{1}(s_t\in\ccalS_i),
      \end{equation}
    and
    \begin{equation}\label{eqn_stochastic_primal}
\hat{\nabla}_{\theta}\ccalL(\theta^k,\lambda^k) = R^{\lambda^k}(\bbs,\bba)\nabla_{\theta}\log \pi_{\theta^k}(a_0|s_0).
      \end{equation}
    In cases where the horizon is finite, the previous expressions can be computed without any additional steps and they yield unbiased estimates of the quantities that they estimate. However, for the infinite horizon case, one would require an infinite trajectory for the later to hold. {An alternative is to consider a finite horizon problem with a random stopping time as discussed in Remark 1, where the stopping time is drawn from a geometric distribution.} By computing the expressions in \eqref{eqn_stochastic_primal} and \eqref{eqn_stochastic_dual} over the randomly drawn horizon the estimates obtained are unbiased \cite{paternain2018stochastic}. The stochastic primal-dual algorithm is summarized in Algorithm \ref{alg_pd}. A different alternative to computing \eqref{eqn_stochastic_primal} is to use actor-critic updates as done in \cite{bhatnagar2012online, tessler2018reward}. Actor-critic methods estimate the gradient with less variance and therefore they enjoy better convergence guarantees. In this work we limit the development to the version of the algorithm given in \eqref{eqn_stochastic_primal} for simplicity.

\section{Numerical Results}\label{num_results}

We have proposed a way to find safe policies via primal-dual methods. In this section, to study the behavior of these proposed methods, we consider a continuous navigation task in an environment filled with hazardous obstacles (See Figure \ref{fig:quiver}). An agent is deployed in this environment and its objective is to reach a goal, while avoiding several obstacles of different size and geometry. On more formal terms, the MDP representing this problem is composed of the state space $\ccalS = [0,10] \times [0,10]$, which represents the position of the agent on the x- and y-axis, i.e., $s=(x,y)$. The agent then takes actions, resulting in its movement along the x- and y-axis. These actions are given by a Gaussian policy, namely
\begin{align}
\pi_\theta(a|s)=\frac{1}{\sqrt{(2 \pi)^2 |\Sigma|}} 
e^{-\frac{1}{2}
\bigl(a - \mu_\theta(s) \bigr)^\top \Sigma^{-1} \bigl(a - \mu_\theta(s) \bigr)
 }
\end{align}
where, we consider a covariance matrix $\Sigma= \diag \left( 0.5,0.5\right)$. Thus, the state of the agent evolves according to the dynamics $s_{t+1}=s_t+T_s a_t$, where $T_s$ is the sampling time of the system, chosen to be $T_s=0.05$. Furthermore, the mean of the Gaussian policy is a function approximator given by a weighted linear combination of Radial Basis Functions (RBF). More specifically,
\begin{align}
\mu_\theta(s)=\sum_{i=1}^{d}  \theta_i 	\exp \left(-\frac{\|s - \bar{s}_i \|^2}{2\sigma^2} \right)
\end{align}
where $\theta=[\theta_1,\ldots,\theta_d]^\top$ is the parameter vector to be learned, $\sigma$ is the bandwidth of each RBF kernel and $\bar{s}_i$ its centers. We choose the bandwidth of the RBF to be $\sigma=0.5$ with centers spaced $0.25$ units from each other.

\begin{figure}[t]
	\centering
	\includegraphics[scale=1]{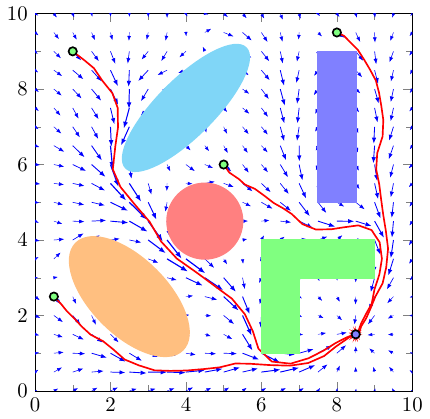} 
	\caption{Navigation policy learned after $40{,}000$ iterations. The agent is trained to navigate to a goal located at $(8.5,1.5)$. Several sample trajectories have been plotted, starting at $(1,9)$, $(8,9.5)$, $(0.5,2.5)$ and $(5,6)$.}
	\label{fig:quiver}
\end{figure}

The specification of the navigation task is given by an infinite horizon problem with discount factor $\gamma=0.95$, and we consider a reward given by the distance between the current state, $s$, and the goal $s_{\texttt{GOAL}}$. Namely,
\begin{align}
r(s,a)=-\|s - s_{\texttt{GOAL}}\|^2,
\end{align}
where, for this specific scenario, we consider the goal to be located at $s_{\texttt{GOAL}}=(8.5,1.5)$. Furthermore, we introduce a constraint into the optimization problem for each of the obstacles, with a demanded level of safety of $1-\delta_i=0.999$ for each of them. We train the agent via the primal-dual algorithm introduced in this work (Algorithm \ref{alg_pd}), where the primal step is performed via the well-known policy gradient \cite{sutton2000policy}. Furthermore, we consider a primal step size $\eta_{\theta}=0.1$ and a dual step size $\eta_{\lambda}=0.05$. We train the policy for $40{,}000$ iterations (until convergence) and obtain the navigation vector flow illustrated in Figure \ref{fig:quiver}.

\begin{figure}[t]
	\centering
	\includegraphics[scale=1]{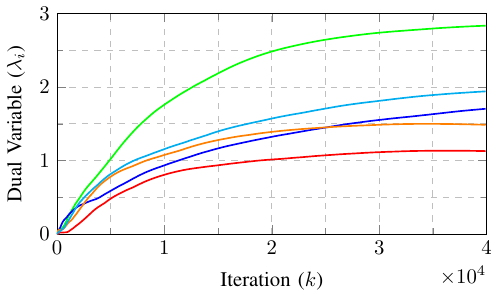} 
	\caption{Average value of the dual variables. The colors of the lines are set accordingly to the constraint they represent as shown in Fig. \ref{fig:quiver}.}
	\label{fig:dualVariables}
\end{figure}

We plot several sample trajectories, illustrating the different possible paths along the learned navigation flow. Overall, given the high level of safety demanded ($1-\delta_i=0.999$), the policy learned by our method avoids all obstacles. Perhaps, a more important observation is that not all obstacles are equally difficult to avoid. By difficult we mean how much the agent must deviate from its straight-line trajectory to the goal to avoid the obstacle. For example, the circular red obstacle in the center of the map appears to be easy to circumnavigate; on the other hand, the green L-shaped obstacle close to the goal appears to require considerable trajectory deviation to be avoided.

In summary, the constraints associated to each of the obstacles are not equally restrictive in the optimization problem. Something that can be further analyzed by taking a closer look at the evolution of the dual variables, which we have considered in Figure \ref{fig:dualVariables}. First, recall that our dual formulation is equivalent to solving a problem with reward
\begin{align}
\label{E:numresreward}
  r_\lambda(s,a) =r(s,a) + \sum_{i=1}^m\lambda_i(\mathbbm{1}(s\in\ccalS_i) -c_i).
\end{align}
where the dual ascent step potentially modifies the values of the dual variables $\lambda_i$ at each iteration, leading to a dynamically changing reward as the algorithm is run. An important observation from the reward \eqref{E:numresreward} is that higher values of the dual variable $\lambda_i$ represent more restrictive constraints (harder to navigate obstacles). In our case, in Figure \ref{fig:dualVariables}, the dual variable associated with the red obstacle has the lowest value, while the opposite is true of the dual variable related to the green obstacle. Hence, in our proposed approach, dual variables allow us to assess the difficulty in satisfying each constraint with its specified level of safety. Clearly, there is great benefit to this more precise analysis (rather than trying to guess purely from observation of the vector field in Figure \ref{fig:quiver}). In our case, it makes it clear which obstacles are worse. While the case between the red and green obstacles might appear clear or apparent to the naked eye, the same might not hold true for other more subtle obstacles. For example, it is not clear which obstacle is the worse between the orange and the cyan one, however looking at the dual variables tells us that the cyan obstacle is slightly worse than the orange one. 

\begin{figure}[t]
	\centering
	\includegraphics[scale=1]{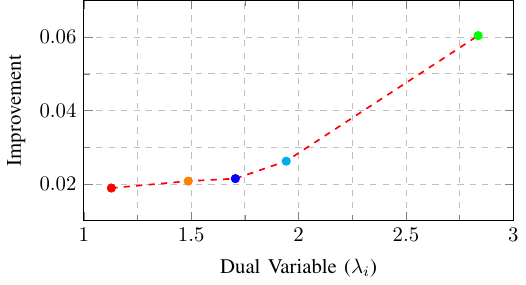} 
	\caption{Improvement in the average reward obtained by removing the obstacle corresponding to the each dual variable and retraining the navigation task. The colors of the points are coordinated with the constraint they represent (cf. Fig. \ref{fig:quiver}). The improvement in the reward is normalized with respect to the number of time steps $T$.}
	\label{fig:obstacleRemoval}
\end{figure}

An immediate consequence of our previous observations and the structure of the reward \eqref{E:numresreward} is the following. If we remove from our system one of the constraints and then retrain the system, obtaining a new policy without the removed constraints, what type of improvement in the reward should we expect? Clearly, we expect larger improvements in the resulting reward if we remove constraints associated with larger dual variables, as they are more restrictive in the optimization problem. We study this in Figure \ref{fig:obstacleRemoval}. As expected, in our navigation scenario, removing the green obstacle provides the larger reward improvement, while the smallest improvement is given by the red obstacle. This is another way of validating what we previously observed: the green obstacle is clearly more cumbersome than the others.

\subsection{Safety guarantees}

\begin{figure}[t]
	\centering
	\includegraphics[scale=1]{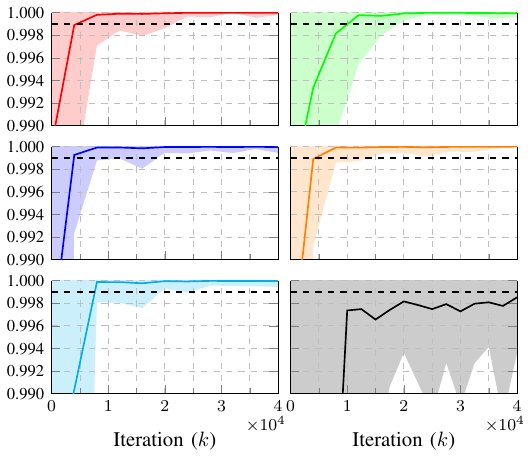} 
	\caption{Evolution of the safety guarantees with respect to the algorithm iterations. Mean values and one standard deviation band have been plotted for $500$ independent evaluations of the policy every $40{,}000$ iterations. The safety requirement $1-\delta_i=0.999$ is shown in dashed lines. Obstacles are color coded according to Fig. \ref{fig:quiver}. The bottom right plot represents the resulting safety guarantees of using a simple policy gradient (without the primal-dual method), where we have naively set the weights of all constraints to the smallest dual variable, corresponding to the red obstacle, $\lambda=1.128$.}
	\label{fig:safetyProb}
\end{figure}

Now, let us take a closer look at the safety guarantees associated with our proposed approach. First, we evaluate the safety attained for each constraint (in our scenario, each obstacle). We look at instances of these probabilities as the algorithm iterates and we consider $500$ independent evaluations of the algorithm at intervals of $40{,}000$ iterations. In Figure \ref{fig:safetyProb}, we plot the resulting mean value and one standard deviation band. Besides the bottom-right plot, we have five different plots, color-coded according to the obstacle they represent (cf. Fig. \ref{fig:quiver}). Recall that the safety requirement is $1-\delta_i=0.999$, which we plot in dashed lines. 

The primal-dual algorithm, as expected, reaches the required safety level first in mean (approximately, after $10{,}000$ iterations, the mean value for the five constraints is under the required safety level). As the algorithm keeps iterating, the resulting distribution gets tighter, and ultimately, the one standard deviation band also falls into the safety requirements. Also, observe that the green obstacle (the one that is harder to navigate around), is the one that takes longer to reach the required level of safety. This obstacle takes longer in mean to reach the required level of safety and has a wider distribution. Furthermore, it is worth remarking that the algorithm converges to a policy which is safer than the one demanded. 

\begin{figure}[t]
	\centering
	\includegraphics[scale=1]{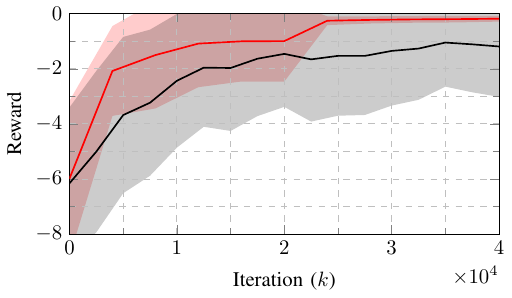} 
	\caption{Evolution of the reward as the algorithm iterates. Mean values and one standard deviation band have been plotted for $500$ independent evaluations. The red line corresponds to the primal-dual algorithm The black line corresponds to setting the weights fixed to the largest dual variable $\lambda=2.835$. The reward is normalized with respect to the number of time steps $T$.}
	\label{fig:reward}
\end{figure}

As an aside, it is reasonable to try to compare our proposed primal-dual approach with a naiver approach. To this end, we consider learning a policy where the weight of the dual variables in the reward \eqref{E:numresreward} are set to a fixed value. Recall that as we discussed previously (cf. equation \eqref{eqn_risk_aware_reward}), this is equivalent to modifying the reward function so as to make it risk-aware. An obvious issue is that, naively, it is not clear which value to use. Nonetheless, we can ask ourselves if we can use one of the values we obtained in Figure \ref{fig:dualVariables} fixed and for all obstacles. We know that if we use the more restrictive value (the one related to the green obstacle), we should be able to maintain the demanded level of safety for all the constraints. However,  what about other values, can we choose, e.g., the lowest value (the one from the red obstacle) and expect to maintain our desired level of safety? The bottom-right plot in Figure \ref{fig:safetyProb}, shows the resulting safety of training the system with a fixed weight of $\lambda=1.128$ (the red obstacle) for all the constraints. Since all the obstacles are set to the same weight, there is a single resulting safety guarantee, the one for all the obstacles. More importantly, the results show that using this weight, does not allow to attain the required level of safety, not in mean, and the standard deviation band does not tighten as before. Clearly, this shows that the choice of the weights used in a reward-shaping approach are not obvious (not even for a simple safe navigation problem as the one we are studying).

Another reasonable approach is to be more restrictive and set all weights fixed to the largest dual variable $\lambda=2.835$ (corresponding to the green obstacle). While setting all variables to the most restrictive value will result in safe policies satisfying the constraints, it has repercussions on the overall reward obtained by the policy. In Figure \ref{fig:reward} we plot the normalized reward (reward per time step) as the policy is trained. We plot the mean reward and one standard deviation band over $500$ independent evaluations of the algorithm at intervals of $40{,}000$ iterations. For the primal-dual algorithm (cf., Algorithm \ref{alg_pd}), the resulting reward converges in around $25{,}000$ iterations of training, a point after which the distribution almost fully converges to the mean. More importantly, the reward of using the naive approach of setting all the weights fixed to the largest dual variable results in a lower overall reward. 

\begin{figure}[t]
	\centering
	\includegraphics[scale=1]{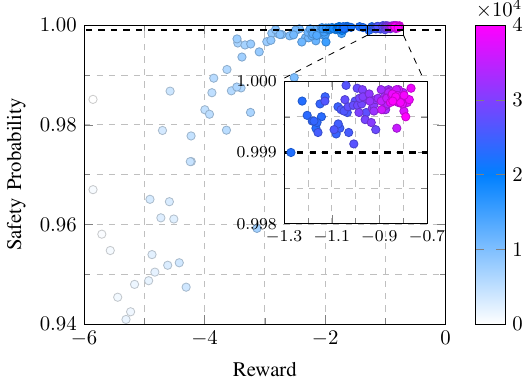} 
	\caption{Resulting reward vs. safety trajectory of the proposed algorithm. The resulting mean reward and safety probability (of the worst constraint) are shown for $500$ independent evaluations of the policy at intervals of $200$ iterations as the policy converges. The iterations of the algorithm are shown as the coloring of the plotted points. The safety requirement $1-\delta_i=0.999$ is shown as a dashed line.}
	\label{fig:trajectoryreward}
\end{figure}

Our previous discussions regarding the results illustrated in Figures \ref{fig:dualVariables} and \ref{fig:reward} highlight the main issues of using classical reward-shaping when attempting to learn safe policies. Namely, it is not clear how to choose the weights. Using weights that are too small will results in policies that are unsafe, while on the other hand, using weights that are too large, will ultimately result in large reductions of the reward attained by the policy. Regarding the latter, using weights that are too large can result in policies that are safe simply because they largely disregard the task reward. In general, classical reward-shaping approaches are not good enough to learn safe policies, as they need to be manually fine-tuned to the problem, which is time and computationally costly. Compared to these approaches, our primal-dual formulation retains some of the benefits of the previous approaches, mainly, part of it resorts to a primal maximization step, which can be computed by many traditional RL methods, such as the well-known policy gradient. More so, our method dynamically chooses the weight via an equivalence between weights and dual variables, in a way that is methodological and guaranteed to attain a good trade-off between safety and reward.

Finally, we plot in Figure \ref{fig:trajectoryreward} the reward-safety trajectory of the primal-dual algorithm (cf., Algorithm \ref{alg_pd}). We show a scatter plot that becomes warmer as the iterations of the algorithm increase. This plot is composed of points relating the mean safety probability (for the worst constraints) against the mean attained reward for $500$ independent evaluations of the policy at intervals of $200$ iterations. As the algorithm is run, both the overall safety and attained reward of the policy increase, until the desired level of safety is attained. Afterwards, the policy attempts to increase the reward while maintaining the desired level of safety. In this case, this shows that early termination of the training of the policy can result in policies that, while being suboptimal, will attain the level of safety required.

\subsection{Complex navigation scenario}

\begin{figure}[t]
	\centering
					
\begin{tikzpicture}
	\pgfplotsset{grid style={opacity=0}}
	\begin{axis}[
		xlabel near ticks,
		ylabel near ticks,
		ticklabel style={font=\small},
		xtick={0,88,...,880},
		ytick={0,88,...,528},
		xticklabels={},
		yticklabels={},
    	width=0.9\columnwidth,
    	height=0.9*0.6\columnwidth,
    	scale only axis,
    	enlargelimits=false,
		grid=both,
		axis on top
	]
	
    \addplot graphics[xmin=0,xmax=880,ymin=0,ymax=528] 
    		{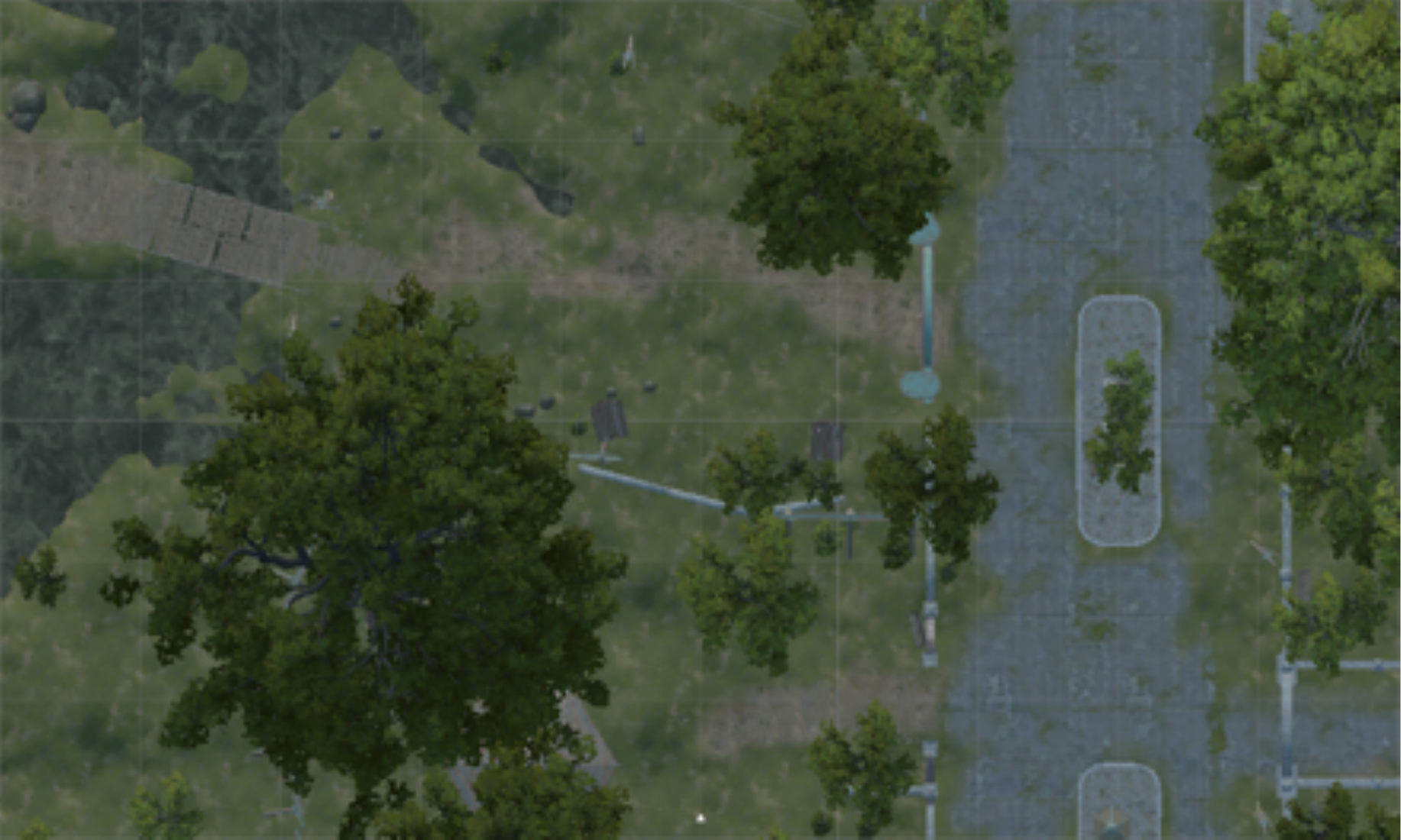};

    \addplot graphics[xmin=0,xmax=880,ymin=0,ymax=528] 
    		{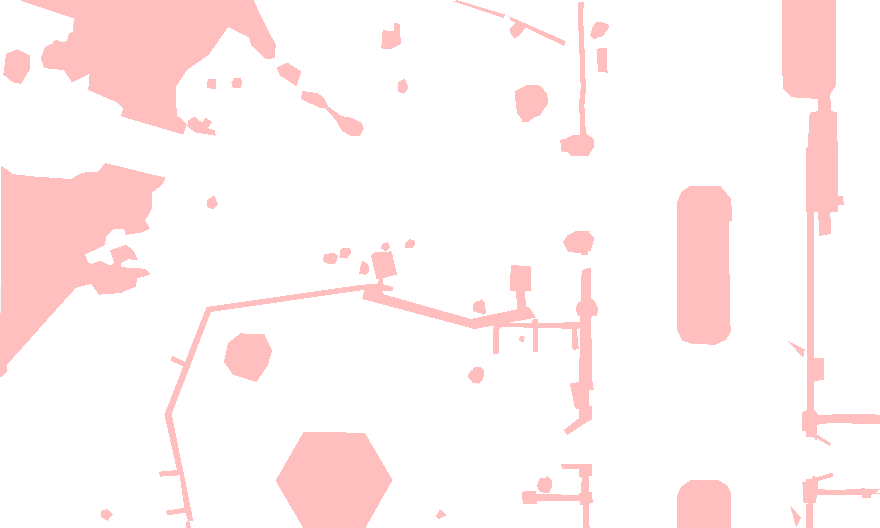};
	
\node[opacity=0.75,scale=0.35, draw, circle, solid, line width=1, minimum height =0.1cm, minimum width =0.1cm, draw=blue!70, fill=blue!50] (Agent1Start) at (880-850,528-150) {};
	
\node[opacity=0.75,scale=0.35, draw, circle, solid, line width=1, minimum height =0.1cm, minimum width =0.1cm, draw=green!70, fill=green!50] (Agent1Final) at (880-20,528-460) {};
	
	\addplot [solid,blue!30, line width=0.7,mark=none,mark options={blue!30,solid}]
		table[x expr=\thisrowno{1},y expr=528-\thisrowno{0}, col sep=comma]{img/dataTrajectoryROSSimulated.dat};

\node[scale=0.05, draw, circle, solid, line width=1, minimum height =0.1cm, minimum width =0.1cm, draw=red!30, fill=red!30] (Agent1Step2) at (29,528-151) {};
\node[scale=0.05, draw, circle, solid, line width=1, minimum height =0.1cm, minimum width =0.1cm, draw=red!30, fill=red!30] (Agent1Step3) at (46,528-150) {};
\node[scale=0.05, draw, circle, solid, line width=1, minimum height =0.1cm, minimum width =0.1cm, draw=red!30, fill=red!30] (Agent1Step4) at (63,528-149) {};
\node[scale=0.05, draw, circle, solid, line width=1, minimum height =0.1cm, minimum width =0.1cm, draw=red!30, fill=red!30] (Agent1Step5) at (82,528-148) {};
\node[scale=0.05, draw, circle, solid, line width=1, minimum height =0.1cm, minimum width =0.1cm, draw=red!30, fill=red!30] (Agent1Step6) at (101,528-145) {};
\node[scale=0.05, draw, circle, solid, line width=1, minimum height =0.1cm, minimum width =0.1cm, draw=red!30, fill=red!30] (Agent1Step7) at (126,528-148) {};
\node[scale=0.05, draw, circle, solid, line width=1, minimum height =0.1cm, minimum width =0.1cm, draw=red!30, fill=red!30] (Agent1Step8) at (155,528-152) {};
\node[scale=0.05, draw, circle, solid, line width=1, minimum height =0.1cm, minimum width =0.1cm, draw=red!30, fill=red!30] (Agent1Step9) at (185,528-159) {};
\node[scale=0.05, draw, circle, solid, line width=1, minimum height =0.1cm, minimum width =0.1cm, draw=red!30, fill=red!30] (Agent1Step10) at (215,528-170) {};
\node[scale=0.05, draw, circle, solid, line width=1, minimum height =0.1cm, minimum width =0.1cm, draw=red!30, fill=red!30] (Agent1Step11) at (247,528-182) {};
\node[scale=0.05, draw, circle, solid, line width=1, minimum height =0.1cm, minimum width =0.1cm, draw=red!30, fill=red!30] (Agent1Step12) at (280,528-194) {};
\node[scale=0.05, draw, circle, solid, line width=1, minimum height =0.1cm, minimum width =0.1cm, draw=red!30, fill=red!30] (Agent1Step13) at (379,528-202) {};
\node[scale=0.05, draw, circle, solid, line width=1, minimum height =0.1cm, minimum width =0.1cm, draw=red!30, fill=red!30] (Agent1Step14) at (437,528-206) {};
\node[scale=0.05, draw, circle, solid, line width=1, minimum height =0.1cm, minimum width =0.1cm, draw=red!30, fill=red!30] (Agent1Step15) at (489,528-220) {};
\node[scale=0.05, draw, circle, solid, line width=1, minimum height =0.1cm, minimum width =0.1cm, draw=red!30, fill=red!30] (Agent1Step16) at (541,528-222) {};
\node[scale=0.05, draw, circle, solid, line width=1, minimum height =0.1cm, minimum width =0.1cm, draw=red!30, fill=red!30] (Agent1Step17) at (586,528-206) {};
\node[scale=0.05, draw, circle, solid, line width=1, minimum height =0.1cm, minimum width =0.1cm, draw=red!30, fill=red!30] (Agent1Step18) at (629,528-213) {};
\node[scale=0.05, draw, circle, solid, line width=1, minimum height =0.1cm, minimum width =0.1cm, draw=red!30, fill=red!30] (Agent1Step19) at (651,528-252) {};
\node[scale=0.05, draw, circle, solid, line width=1, minimum height =0.1cm, minimum width =0.1cm, draw=red!30, fill=red!30] (Agent1Step20) at (656,528-291) {};
\node[scale=0.05, draw, circle, solid, line width=1, minimum height =0.1cm, minimum width =0.1cm, draw=red!30, fill=red!30] (Agent1Step21) at (663,528-330) {};
\node[scale=0.05, draw, circle, solid, line width=1, minimum height =0.1cm, minimum width =0.1cm, draw=red!30, fill=red!30] (Agent1Step22) at (683,528-370) {};
\node[scale=0.05, draw, circle, solid, line width=1, minimum height =0.1cm, minimum width =0.1cm, draw=red!30, fill=red!30] (Agent1Step23) at (719,528-397) {};
\node[scale=0.05, draw, circle, solid, line width=1, minimum height =0.1cm, minimum width =0.1cm, draw=red!30, fill=red!30] (Agent1Step24) at (805,528-450) {};
\node[scale=0.05, draw, circle, solid, line width=1, minimum height =0.1cm, minimum width =0.1cm, draw=red!30, fill=red!30] (Agent1Step25) at (840,528-459) {};
\node[scale=0.05, draw, circle, solid, line width=1, minimum height =0.1cm, minimum width =0.1cm, draw=red!30, fill=red!30] (Agent1Step26) at (860,528-455) {};

\draw[->, solid, red!30, line width=0.05, >=latex] (Agent1Step2) to (Agent1Step3);	
\draw[->, solid, red!30, line width=0.05, >=latex] (Agent1Step3) to (Agent1Step4);	
\draw[->, solid, red!30, line width=0.05, >=latex] (Agent1Step4) to (Agent1Step5);	
\draw[->, solid, red!30, line width=0.05, >=latex] (Agent1Step5) to (Agent1Step6);	
\draw[->, solid, red!30, line width=0.05, >=latex] (Agent1Step6) to (Agent1Step7);	
\draw[->, solid, red!30, line width=0.05, >=latex] (Agent1Step7) to (Agent1Step8);	
\draw[->, solid, red!30, line width=0.05, >=latex] (Agent1Step8) to (Agent1Step9);	
\draw[->, solid, red!30, line width=0.05, >=latex] (Agent1Step9) to (Agent1Step10);	
\draw[->, solid, red!30, line width=0.05, >=latex] (Agent1Step10) to (Agent1Step11);	
\draw[->, solid, red!30, line width=0.05, >=latex] (Agent1Step11) to (Agent1Step12);	
\draw[->, solid, red!30, line width=0.05, >=latex] (Agent1Step12) to (Agent1Step13);	
\draw[->, solid, red!30, line width=0.05, >=latex] (Agent1Step13) to (Agent1Step14);	
\draw[->, solid, red!30, line width=0.05, >=latex] (Agent1Step14) to (Agent1Step15);	
\draw[->, solid, red!30, line width=0.05, >=latex] (Agent1Step15) to (Agent1Step16);	
\draw[->, solid, red!30, line width=0.05, >=latex] (Agent1Step16) to (Agent1Step17);	
\draw[->, solid, red!30, line width=0.05, >=latex] (Agent1Step17) to (Agent1Step18);	
\draw[->, solid, red!30, line width=0.05, >=latex] (Agent1Step18) to (Agent1Step19);	
\draw[->, solid, red!30, line width=0.05, >=latex] (Agent1Step19) to (Agent1Step20);	
\draw[->, solid, red!30, line width=0.05, >=latex] (Agent1Step20) to (Agent1Step21);	
\draw[->, solid, red!30, line width=0.05, >=latex] (Agent1Step21) to (Agent1Step22);	
\draw[->, solid, red!30, line width=0.05, >=latex] (Agent1Step22) to (Agent1Step23);	
\draw[->, solid, red!30, line width=0.05, >=latex] (Agent1Step23) to (Agent1Step24);	
\draw[->, solid, red!30, line width=0.05, >=latex] (Agent1Step24) to (Agent1Step25);	
\draw[->, solid, red!30, line width=0.05, >=latex] (Agent1Step25) to (Agent1Step26);	
	
\end{axis}
\end{tikzpicture}
	\caption{Navigation of a complex agent in a realistic simulation. Major obstacles are highlighted in red. The agent starts at the blue position and navigates to the green goal. The ideal trajectory for the trained agent are illustrated by the red arrows, while the real trajectory experienced by the agent in the ROS simulation is shown in blue. The scale of the environment is $100 \times 60$ meters, with each line submark corresponding to $10$ meters.}
	\label{fig:complex_environment}
\end{figure}

We now consider the application of our proposed methodology to a more complex and realistic scenario. To do so, we perform simulations in the Robot Operating System (ROS) supported by the Unity physics engine. The environment that we consider for these simulations is shown in Fig. \ref{fig:complex_environment}. In this map, the agent must navigate from the blue dot in the top left corner to the green goal on the bottom right corner while avoiding the obstacles highlighted in red. The agent used for this task is a Clearpath Husky, a differential drive wheeled robot, whose physics interactions with the environment are realistically simulated by our ROS/Unity setup. 

The state space of the agent consists of its x- and y-axis positions together with its heading angle $\psi$. Furthermore, the agent takes actions corresponding to its linear $ v$ and angular $\dot \psi$ velocities, and the model has a sampling time $T_s$ which is set to $T_s=0.2$.  The system has the following dynamics
\begin{align}
%s_{t+1}=
\begin{bmatrix}
x_{t+1} \\
y_{t+1} \\
\psi_{t+1}
\end{bmatrix}
=
\begin{bmatrix}
x_{t} \\
y_{t} \\
\psi_{t}
\end{bmatrix}
+
T_s
\begin{bmatrix}
{v}_t \cos(\psi_t) \\
{v}_t \sin(\psi_t)  \\
\dot{\psi_t}
\end{bmatrix}.
\end{align}

In order to operate the Husky platform, we aim to learn a drive controller producing linear and angular velocities using Algorithm \ref{alg_pd} with the same parameters as in the previous numerical results. The agent is trained to avoid the obstacles, highlighted in red in Fig. \ref{fig:complex_environment}. Then, once the policy has been learned, it is deployed in the realistic ROS simulator and evaluated with high-fidelity Unity physics to obtain a more realistic portrayal of the performance of our policy.

Figure \ref{fig:complex_environment} illustrates the resulting trajectories, where the ideal trained trajectory corresponds to the red arrows and the trajectory followed by the agent in the ROS simulation is shown in blue. First, we verify that the agent successfully reaches the goal while avoiding all obstacles. Indeed, in the realistic ROS simulation it does so while following closely the ideal trained trajectory. In this regard, we can more carefully observe how the agent starts close to a bridge surrounded by water (obstacles). This corresponds to a very risky region as it is easy to run into obstacles. Thus, to maintain safety, the agent keeps a very low velocity which allows it to stay close to the ideal trajectory. Once the agent crosses the bridge, it increases its velocity as the region is devoid of major obstacles and it is safe to do so. As it crosses the arch into the main road, the agent turns and minimally reduces the speed to not collide with the central obstacle of the paved road. However, we see how the realistic ROS simulation deviates slightly from the ideal trajectory at high velocity turns. The trained model does not account for the differential drive of the Husky platform, which has a hard time turning at high speed and which is accounted for in the ROS simulation. Nonetheless, the agent is capable is staying close to the ideal trajectory, safely navigating to the goal.

%!TEX root = root.tex

\section{Conclusions}\label{sec_conclusions}
In this paper, we have studied the problem of learning safe policies in RL problems. More specifically, we have introduced safety into the problem as a probabilistic version of positive invariance. These constraints are then relaxed to provide expressions that are algorithmically useful, i.e., we can compute gradients of the constraints with respect to the parameters of the policy. These relaxations are such that they guarantee that the safety requirements are preserved for finite horizon operation and until an arbitrarily large horizon for discounted infinite horizon problems.  In addition, we show that the relaxed problem has a duality gap that can be made arbitrarily small and thus primal-dual algorithms are suitable to solve the safe RL problem. Numerical results for an agent navigating a world filled with hazardous obstacles, in a complex and realistic environment, show that the proposed scheme dynamically adapts the cost of safety. Compared to previous approaches, our proposed scheme provides safe policies with guarantees and a systematic way of achieving them, without being reliant on the manual tuning of parameters. 

%through probabilistic constraints that we then relax for both finite and infinite horizons, hence formulating a constrained optimization problem. The advantages of the proposed relaxations are threefold. First, they allow us to compute the primal and dual gradients of the Lagrangian associated to the optimization problem, which can be solved by running a stochastic primal-dual method. Second, the relaxed problem has a duality gap that can be made arbitrarily small and therefore the solution computed through the primal-dual algorithm is ensured to be optimal. Finally, these relaxations do not come at the cost of safety. In particular, we established that the finite horizon problem remains safe and we established a safe horizon for the discounted optimization problem. 

\appendix
%!TEX root = root.tex

\section{Appendix}

\subsection{Proof of Lemma \ref{T:prob_inequality}}\label{proof_lemma_safety}

We proceed by a sequence of implications. Notice that 
\begin{equation}
	\sum_{t = 0}^{\infty} \mu_t \mathbb{P}(\ccalE_t)
		= \sum_{t = 0}^{\infty} \mu_t - \sum_{t = 0}^{\infty} \mu_t \mathbb{P}(\bar{\ccalE_t})
		\text{,}
\end{equation}
where~$\bar{\ccalE_t}$ denotes the complement of $\ccalE_t$, for all $t\geq 0$. Hence,
\begin{equation}
	\sum_{t = 0}^{\infty} \mu_t \mathbb{P}(\ccalE_t) \geq \sum_{t = 0}^{\infty} \mu_t - \mu^\prime \delta
	\Rightarrow
	\sum_{t = 0}^{\infty} \mu_t \mathbb{P}(\bar{\ccalE_t}) \leq \mu^\prime \delta
		\text{.}
\end{equation}
Since all summands are non-negative, we can lower bound the right-hand side by truncating the summation at~$T$, obtaining
\begin{equation}
	\sum_{t = 0}^{\infty} \mu_t \mathbb{P}(\bar{\ccalE_t}) \leq \mu^\prime \delta
	\Rightarrow
	\sum_{t = 0}^{T} \mu_t \mathbb{P}(\bar{\ccalE_t}) \leq \mu^\prime \delta
		\text{.}
\end{equation}
Since the~$\mu_t$ are non-increasing, it holds that~$\sum_{t = 0}^{T} \mu_t \mathbb{P}(\bar{\ccalE_t}) \geq \mu_T \sum_{t = 0}^{T} \mathbb{P}(\bar{\ccalE_t})$. Then, using the fact that~$\mu^\prime \leq \mu_T$ yields
\begin{equation}
	\sum_{t = 0}^{T} \mu_t \mathbb{P}(\bar{\ccalE_t}) \leq \mu^\prime \delta
	\Rightarrow
	\sum_{t = 0}^{T} \mathbb{P}(\bar{\ccalE_t}) \leq \delta
		\text{.}
\end{equation}
To conclude, we can apply the Boole-Fr\'{e}chet-Bonferroni inequality~\cite[1.6.10]{durrett2010probability} to obtain
\begin{equation}
	\sum_{t = 0}^{T} \mathbb{P}(\bar{\ccalE_t}) \leq \delta
	\Rightarrow
	\mathbb{P}\left( \bigcap_{t = 0}^{T} \ccalE_t \right) \geq 1 - \delta
		\text{.}\qedhere
\end{equation}
%

%%%%%%%%%%%%%%%%%%%%%%%%%%%%%%%%%%%%%%%%%%%%%%%%%%
%%%%% ZERO DUALITY GAP
%%%%%%%%%%%%%%%%%%%%%%%%%%%%%%%%%%%%%%%%%%%%%%%%%%
{
\subsection{Proof of Lemma \ref{L:convexity}}\label{proof_convexity}
Denote by $\boldsymbol{0}_m\in\mathbb{R}^m$ the null vector. Notice that $(P^\star,\boldsymbol{0}_m^\top)^\top \in \ccalC$. Indeed, for the policy $\pi^\star$ that solves \eqref{P:constrainedRL} we have that $V(\pi^\star)=P^\star$ and that $U_i(\pi^\star)-c_i\geq 0$. Thus, $\ccalC$ is not empty.

Let $\xi^1, \xi^2 \in \ccalC$ and let $\pi^1,\pi^2\in \ccalP_{s}(\ccalA)$ be respective policies that satisfy $V(\pi^\ell) \geq \xi^\ell_0$ and $U_i(\pi^\ell) \geq \xi^\ell_i$ for all $i=1,\ldots m$ and $\ell = \left\{1,2\right\}$. To show that $\ccalC$ is convex it suffices to establish that for any $\mu\in(0,1)$, it holds that $\mu\xi^1+(1-\mu)\xi^2 \in\ccalC$. 

To do so we start by writing the value function $V(\pi)$ in the following form (see e.g., \cite{borkar1988convex})
\begin{equation}\label{eqn_v_occupation}
(1-\gamma)V(\pi) = \int_{\ccalS\times\ccalA} r(s,a)\rho(s,a) \, ds da.
\end{equation}
where $\rho(s,a)$ is the occupation meassure induced by policy $\pi$ (cf., \eqref{eqn_occupation_measure}). Likewise we can write
\begin{equation}
(1-\gamma)U_i(\pi) = \int_{\ccalS\times\ccalA} \mathbbm{1}(s\in\ccalS_i) \rho(s,a)\, ds da.
\end{equation}
Denote by $\ccalM(\ccalS,\ccalA)$ the measures over $\ccalS\times\ccalA$ and define the set $\ccalR$ as the set of all occupation measures induced by the policies $\pi\in\ccalP_s(\ccalA)$ as
\begin{equation}
\begin{split}
\ccalR & \triangleq \left\{\rho\in\ccalM(\ccalS,\ccalA)\Big| \right. \\
&  \left.\rho(s,a)=\left(1-\gamma\right)\left(\sum_{t=0}^\infty \gamma^tp_\pi(s_t=s,a_t=a)\right) \right\}.
\end{split}
\end{equation}
It follows from \cite[Theorem 3.1]{borkar1988convex} that the set of occupation measures $\ccalR$ is convex and compact. Let $\rho^1,\rho^2\in\ccalR$ be the occupancy measures associated to the policies $\pi^1,\pi^2 \in \ccalP_s(\ccalA)$. Since the set is compact, there exists $\pi^\mu\in\ccalP_s(\ccalA)$ such that $\mu \rho^1+(1-\mu)\rho^2 \in \ccalR$. Using the linearity of the integral it follows that 
\begin{align}
V(\pi^\mu) &= \frac{\mu}{1-\gamma} \int_{\ccalS\times\ccalA} r(s,a)\rho^1(s,a) \, ds da  \nonumber \\
&+\frac{1-\mu}{1-\gamma}\int_{\ccalS\times\ccalA} r(s,a)\rho^2(s,a) \, ds da.
\end{align}
Using the expression in \eqref{eqn_v_occupation} it follows that $V(\pi^\mu) = \mu V(\pi^1) +(1-\mu)V(\pi^2)$. Which shows that $V(\pi^\mu) \geq \mu\xi^1_0+(1-\mu)\xi^2_0$. Using the same arguments we can establish that for every $i=1,\ldots m$ we have that $U_i(\pi^\mu) \geq \mu\xi^1_i+(1-\mu)\xi^2_i$. This completes the proof that for any $\xi^1,\xi^2\in\ccalC$ and for any $\mu \in(0,1)$ the convex combination is also contained in $\ccalC$, i.e., $\mu\xi^1+(1-\mu)\xi^2\in\ccalC$. Thus, the set $\ccalC$ is convex.

}

%%%%%%%%%%%%%%%%%%%%%%%%%%%%%%%%%%%%%%%%%%%%%%%%%%
%%%%% GEOMETRIC SERIES LEMMA
%%%%%%%%%%%%%%%%%%%%%%%%%%%%%%%%%%%%%%%%%%%%%%%%%%
%
\subsection{Proof of Lemma~\ref{lemma_meassure_bound}}\label{proof_meassure_bound}
  Let us start by writing the left hand side of \eqref{eqn_occupation_bound} as
  \begin{equation}\label{aux_measure_aux}
    \begin{split}
      &\int_{\ccalS\times\ccalA} |\rho(s,a)-\rho_\theta(s,a)| \,dsda  \\
      &= (1-\gamma) \int_{\ccalS\times\ccalA} \left|\sum_{t=0}^\infty \gamma^t\left( p^t_\pi(s,a)-p_\theta^t(s,a)\right)\right| \,dsda.
\end{split}
    \end{equation}
  Using the triangle inequality, we upper bound \eqref{aux_measure_aux} as
    \begin{equation}\label{eqn_aux_occupation_bound1}
\begin{split}
  &\int_{\ccalS\times\ccalA} |\rho(s,a)-\rho_\theta(s,a)| \,dsda  \\
  &\leq (1-\gamma) \sum_{t=0}^\infty \gamma^t\int_{\ccalS\times\ccalA} \left| p_\pi^t(s,a)-p^t_\theta(s,a)\right| \,dsda .
\end{split}    \end{equation}
    Notice that to complete the proof it suffices to show that the right hand side of the previous expression is bounded by $\epsilon/(1-\gamma)$. We next work towards that end, and we start by bounding the difference $\left| p_\pi^t(s,a)-p_\theta^t(s,a)\right|$. Notice that this difference can be upper bounded using the triangle inequality as
    \begin{align}\label{eqn_aux_bound_diff_p}
      \left| p_\pi^t(s,a)-p_\theta^t(s,a)\right| &\leq p_\pi^t(s)\left| \pi(a|s)-\pi_\theta(a|s)\right| \nonumber\\
      &+\pi_\theta(a|s)\left|p_\pi^t(s)-p_\theta^t(s)\right|.
  \end{align}
It follows from Definition \ref{def_universal_param} that
    \begin{equation}\label{eqn_aux_occupation_bound2}
\int_{\ccalS\times\ccalA} p_\pi^t(s)\left| \pi(a|s)-\pi_\theta(a|s)\right|\, ds da \leq\epsilon \int_{\ccalS} p_\pi^t(s)\, ds = \epsilon,
      \end{equation}
    where the last equality follows from the fact that $p_\pi^t(s)$ is a density and thus integrates to one. We next work towards bounding the integral of the second term in \eqref{eqn_aux_bound_diff_p}. Using the fact that $\pi_\theta(a|s)$ is a density, it follows that 
    \begin{equation}
\int_{\ccalS\times\ccalA} \pi_\theta(a|s)\left|p_\pi^t(s)-p_\theta^t(s)\right|\, ds da = \int_{\ccalS} \left|p_\pi^t(s)-p_\theta^t(s)\right|\, ds.
      \end{equation}
        Notice that the previous difference is zero for $t=0$ and for any $t>0$ it can be upper bounded by
                \begin{align}\label{eqn_aux_occupation_bound3}
                  &\int_{\ccalS} \left|p_\pi^t(s)-p_\theta^t(s)\right|\, ds \nonumber \\
                  &\leq \int_{\ccalS} \int_{\ccalS\times\ccalA} p(s|s^\prime,a^\prime)\left|p_\pi^{t-1}(s^\prime,a^\prime)-p_\theta^{t-1}(s^\prime,a^\prime)\right|\, ds ds^\prime da^\prime \nonumber \\
& = \int_{\ccalS\times\ccalA} \left|p_\pi^{t-1}(s^\prime,a^\prime)-p_\theta^{t-1}(s^\prime,a^\prime)\right|\, ds^\prime da^\prime
      \end{align}
Combining the bounds in \eqref{eqn_aux_occupation_bound2}, \eqref{eqn_aux_occupation_bound2} and \eqref{eqn_aux_occupation_bound3} yields
\begin{align}
  &  (1-\gamma) \sum_{t=0}^\infty \gamma^t\int_{\ccalS\times\ccalA} \left|\left( p_\pi^t(s,a)-p_\theta^t(s,a)\right)\right| \,dsda \leq \nonumber \\ 
&(1-\gamma) \sum_{t=1}^\infty \gamma^t\int_{\ccalS\times\ccalA} \left|\left( p_\pi^{t-1}(s,a)-p_\theta^{t-1}(s,a)\right)\right| \,ds da \nonumber \\ 
  &+  (1-\gamma) \sum_{t=0}^\infty \gamma^t\epsilon
    \end{align}
Notice that the second term in the right hand side of the previous expression is the sum of the geometric multiplied by $1-\gamma$. Hence we have that $ (1-\gamma) \sum_{t=0}^\infty \gamma^t\epsilon=\epsilon $. The first term in the right hand side of the previous expression is in fact the same as the term in the left hand side of the expression multiplied by the discount factor $\gamma$. Thus, rearranging the terms, the previous expression implies that 
\begin{equation}
  (1-\gamma) \sum_{t=0}^\infty \gamma^t\int_{\ccalS\times\ccalA} \left|\left( p_\pi^t(s,a)-p_\theta^t(s,a)\right)\right| \,dsda \leq \frac{\epsilon}{1-\gamma}.
\end{equation}
This completes the proof of the Lemma.

%%%%%%%%%%%%%%%%%%%%%%%%%%%%%%%%%%%%%%%%%%%%%%%%%%
%%%%% PARAMETRIZATION THEOREM
%%%%%%%%%%%%%%%%%%%%%%%%%%%%%%%%%%%%%%%%%%%%%%%%%%
\subsection{Proof of Theorem~\ref{T:parametrization}}\label{proof_param}
The lower bound on $D_\theta^\star$ follows directly from weak duality \cite[Chapter 5]{boyd2004convex}. To prove the upper bound notice that the dual functions $d(\lambda)$ and $d_\theta(\lambda)$ associated to the problems \eqref{P:constrainedRL} and \eqref{P:parametric} respectively are such that for every $\lambda \in \mathbb{R}^m_+$ we have that $  d_\theta(\lambda) \leq d(\lambda)$. The latter follows from the fact that the set of maximizers of the Lagrangian for the parameterized policies is contained in the set of maximizers of stochastic kernels. In particular, this holds for $\lambda^\star$ the solution to \eqref{P:dualRL}. Hence we have the following sequence of inequalities
  \begin{equation}
D^\star = d(\lambda^\star) \geq d_{\theta}(\lambda^\star) \geq D^\star_\theta, 
    \end{equation}
  where the last inequality follows from the fact that $D^\star_\theta$ is the minimum of \eqref{P:dual_parametric}. The zero duality gap established in Theorem \ref{T:strongDuality} allows us to upper bound $D_\theta^\star \leq P^\star$, where $P^\star $ is the value of the problem \eqref{P:constrainedRL}. To complete the proof of the result we are left to show that $P^\star \leq P^\star_\theta +  \left(B_{r}+\left\|\lambda_\epsilon^\star\right\|_1\right) \epsilon/(1-\gamma)$. Notice that if the perturbed problem \eqref{P:perturbed} with $\xi_i = \epsilon/(1-\gamma)$ is infeasible, then the dual optimum is infinite and the result holds trivially. Hence we focus on the case where \eqref{P:perturbed} is feasible. We claim the following three facts to be proven. (i) Let $\pi^\star_\epsilon$ be the solution to the perturbed problem \eqref{P:perturbed} with $\xi_i= \epsilon/(1-\gamma)$ for all $i = 1,\ldots, m$. Then, 
\begin{equation}
P^\star \leq P_\epsilon^\star + \left\|\lambda_\epsilon^\star\right\|_1 \frac{\epsilon}{1-\gamma} = V(\pi^\star_\epsilon)+ \left\|\lambda_\epsilon^\star\right\|_1 \frac{\epsilon}{1-\gamma}
\end{equation}
(ii) Let $\pi^\star_{\epsilon,\theta}$ be an element of the $\epsilon$-parameterization that approximates $\pi^\star_\epsilon$ in the sense of Definition \ref{def_universal_param}. Then
\begin{equation}
V(\pi^\star_\epsilon) \leq V(\pi^\star_{\epsilon,\theta}) + \frac{B_r \epsilon}{1-\gamma}. 
\end{equation}
(iii) Lastly we claim that $\pi^\star_{\epsilon,\theta}$ is feasible for \eqref{P:parametric}. And therefore $V(\pi^\star_{\epsilon,\theta})\leq P^\star_\theta$.  These three facts complete the proof. 

We start by showing (i).  Similar to the proof of Theorem \ref{T:strongDuality}, we can define the following set
\begin{align}\label{eqn_convex_set2}
\ccalC_\epsilon = \bigl\{\xi \in \mathbb{R}^{m+1}\big| & \exists \pi \in \ccalP_{s}(\ccalA), V(\pi) \geq \xi_0, \nonumber\\
&U_i(\pi) \geq c_i+\frac{\epsilon}{1-\gamma}+\xi_i, \forall i=1,\ldots, m \bigr\}.
\end{align} 
Analogous to the proof of Lemma \ref{L:convexity} the previous set is convex and non-empty. In addition, just as we did in the proof of Theorem \ref{T:strongDuality} we have that $(P^\star_\epsilon, \boldsymbol{0}_m^\top)^\top$ is at the boundary of $\ccalC_\epsilon$ and that the separating hyperplane at that point is defined by $\tilde{\lambda} = (1,({\lambda}^\star_\epsilon)^\top)^\top$.  Hence, for any $\xi\in\ccalC_\epsilon$ it follows that 
\begin{equation}
		(P_\epsilon^\star,\boldsymbol{0}_m^\top) \tilde{\lambda}\geq \xi^\top \tilde{\lambda}.
		\end{equation}
{In particular, notice that the the point $(P^\star, -\mathbbm{1}_m\frac{\epsilon}{1-\gamma})\in\ccalC_\epsilon$. Indeed, for $\pi^\star$ optimal policy of \eqref{P:constrainedRL} we have that $V(\pi^\star) = P^\star$ and $U_i(\pi^\star) \geq c_i $.  } Thus, substituting in the previous equation $\xi = (P^\star, -\mathbbm{1}_m\frac{\epsilon}{1-\gamma})$ and $\tilde{\lambda} = (1,({\lambda}^\star_\epsilon)^\top)^\top$ it follows that 
\begin{equation}
		P_\epsilon^\star \geq P^\star -\frac{\epsilon}{1-\gamma} \left\|\lambda^\star_\epsilon\right\|_1.
		\end{equation}
This completes the proof of the claim (i). 

}

To establish (ii) notice that we can write the difference $V(\pi^\star_\epsilon)-V(\pi^\star_{\epsilon,\theta})$ as 
\begin{equation}
V(\pi^\star_\epsilon)-V(\pi^\star_{\epsilon,\theta}) = \int_{\ccalS\times\ccalA} r(s,a) \left( d\rho^\star_\epsilon -d\rho_{\epsilon,\theta}^\star\right),
\end{equation}
where in the previous expression $\rho_\epsilon^\star$ and $\rho_{\epsilon,\theta}^\star$ are the occupancy measures corresponding to the policies $\pi^\star_\epsilon$ and $\pi^\star_{\theta,\epsilon}$ respectively (cf., \eqref{eqn_occupation_measure}).  Since $\pi_{\epsilon,\theta^\star}$ is by definition an $\epsilon$ approximation of $\pi_\epsilon^\star$ it follows from Lemma \ref{lemma_meassure_bound} that
    \begin{equation}
      \int_{\ccalS\times\ccalA} \left|d \rho^\star_\epsilon-d\rho_{\epsilon,\theta}^\star(\lambda)\right|  \leq \frac{\epsilon}{1-\gamma}.
    \end{equation}
Using the fact that the rewards are bounded by $B_r$ completes the proof of (ii). 

The proof of (iii) is analogous to that of (ii) but reasoning with the constraints functions $U_i$.

\bibliographystyle{ieeetr}
\bibliography{bib,ml}
\begin{IEEEbiography}[{\includegraphics[width=1in,height=1.25in,clip,keepaspectratio]{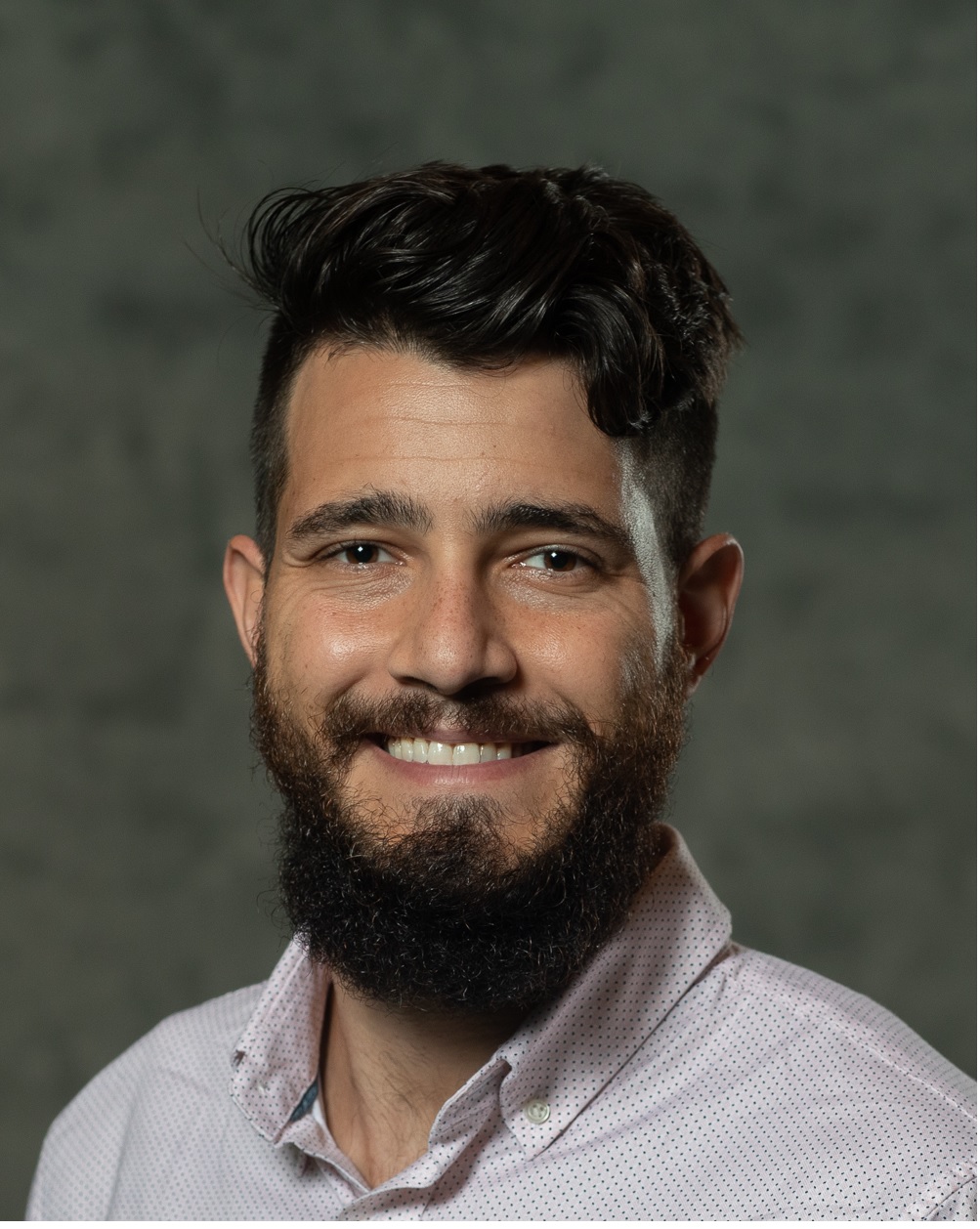}}]
  {Santiago Paternain} received the B.Sc. degree in electrical engineering from Universidad de la República Oriental del Uruguay, Montevideo, Uruguay in 2012, the M.Sc. in Statistics from the Wharton School in 2018 and the Ph.D. in Electrical and Systems Engineering from the Department of Electrical and Systems Engineering, the University of Pennsylvania in 2018. He is currently an Assistant Professor in the Department of Electrical Computer and Systems Engineering at the Rensselaer Polytechnic Institute. Prior to joining Rensselaer, Dr. Paternain was a postdoctoral Researcher at the University of Pennsylvania. His research interests lie at the intersection of machine learning and control of dynamical systems. Dr. Paternain was the recipient of the 2017 CDC Best Student Paper Award and the 2019 Joseph and Rosaline Wolfe Best Doctoral Dissertation Award from the Electrical and Systems Engineering Department at the University of Pennsylvania. 
\end{IEEEbiography}
\begin{IEEEbiography}[{\includegraphics[width=1in,height=1.25in,clip,keepaspectratio]{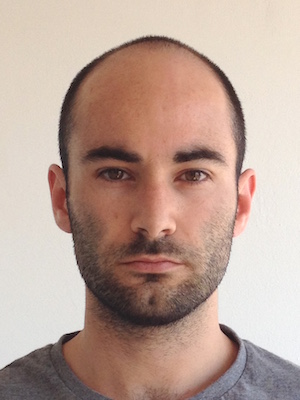}}]{Miguel Calvo-Fullana} received his B.Sc. degree in electrical engineering from the Universitat de les Illes Balears (UIB), in 2012 and the M.Sc. and Ph.D. degrees in electrical engineering from the Universitat Polit\`ecnica de Catalunya (UPC), in 2013 and 2017, respectively. From September 2012 to July 2013 he was a research assistant with Nokia Siemens Networks (NSN) and Aalborg University (AAU). From December 2013 to July 2017, he was with the Centre Tecnol\`ogic de Telecomunicacions de Catalunya (CTTC) as a research assistant. From September 2017 to September 2020, he was a postdoctoral researcher at the University of Pennsylvania. Since September 2020, he is a postdoctoral researcher at the Massachusetts Institute of Technology. His research interests lie in the broad areas of learning and optimization for autonomous systems. In particular, he is interested in multi-robot systems with an emphasis on wireless communication and network connectivity.
 \end{IEEEbiography}
\begin{IEEEbiography}[{\includegraphics[width=1in,height=1.25in,clip,keepaspectratio]{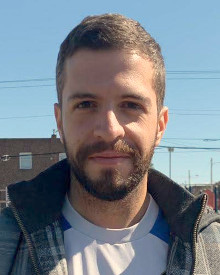}}]{Luiz F.O. Chamon}~(S’12, M'21) received the B.Sc. and M.Sc. degrees in electrical engineering from the University of São Paulo, São Paulo, Brazil, in 2011 and 2015 and the Ph.D. degree in electrical and systems engineering from the University of Pennsylvania (Penn), Philadelphia, in 2020. He is currently a postdoc at the Simons Institute of the University of California, Berkeley. In 2009, he was an undergraduate exchange student of the Masters in Acoustics of the École Centrale de Lyon, Lyon, France, and worked as an Assistant Instructor and Consultant on nondestructive testing at INSACAST Formation Continue. From 2010 to 2014, he worked as a Signal Processing and Statistics Consultant on a research project with EMBRAER. In 2018, he was recognized by the IEEE Signal Processing Society for his distinguished work for the editorial board of the IEEE Transactions on Signal Processing. He also received both the best student paper and the best paper awards at IEEE ICASSP 2020. His research interests include optimization, signal processing, machine learning, statistics, and control.

\end{IEEEbiography}
\begin{IEEEbiography}[{\includegraphics[width=1in,height=1.25in,clip,keepaspectratio]{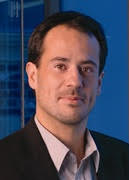}}]{Alejandro Ribeiro}  received the B.Sc. degree in electrical engineering from the Universidad de la Republica Oriental del Uruguay, Montevideo, in 1998 and the M.Sc. and Ph.D. degree in electrical engineering from the Department of Electrical and Computer Engineering, the University of Minnesota, Minneapolis in 2005 and 2007. From 1998 to 2003, he was a member of the technical staff at Bellsouth Montevideo. After his M.Sc. and Ph.D studies, in 2008 he joined the University of Pennsylvania (Penn), Philadelphia, where he is currently Professor of Electrical and Systems Engineering. His research interests are in the applications of statistical signal processing to collaborative intelligent systems. His specific interests are in wireless autonomous networks, machine learning on network data and distributed collaborative learning. Papers coauthored by Dr. Ribeiro received the 2014 O. Hugo Schuck best paper award, and paper awards at CDC 2017, SSP Workshop 2016, SAM Workshop 2016, Asilomar SSC Conference 2015, ACC 2013, ICASSP 2006, and ICASSP 2005. His teaching has been recognized with the 2017 Lindback award for distinguished teaching and the 2012 S. Reid Warren, Jr. Award presented by Penn’s undergraduate student body for outstanding teaching. Dr. Ribeiro is a Fulbright scholar class of 2003 and a Penn Fellow class of 2015.
\end{IEEEbiography}

%\newpage

%\input{proof_of_proposition_zero}

\end{document}